\documentclass{amsart}

\newtheorem{theorem}{Theorem}[section]
\newtheorem{lemma}[theorem]{Lemma}

\newtheorem{corollary}[theorem]{Corollary}

\theoremstyle{definition}
\newtheorem{definition}[theorem]{Definition}
\newtheorem{example}[theorem]{Example}

\theoremstyle{remark}

\newcommand{\R}{{\mathbb R}}
\newcommand{\Z}{{\mathbb Z}}
\newcommand{\C}{{\mathbb C}}

\begin{document}

\title[Symmetry preserving self-adjoint extensions]{Symmetry preserving self-adjoint extensions of Schr\"odinger operators with singular potentials}

\author{D.M.~Gitman}
\address{Institute of Physics, University of Sao Paulo, Brazil}
\email{gitman@dfn.if.usp.br}
\author{A.G.~Smirnov}
\address{I.~E.~Tamm Theory Department, P.~N.~Lebedev
Physical Institute, Leninsky prospect 53, Moscow 119991, Russia}

\email{smirnov@lpi.ru}

\author{I.V.~Tyutin}
\address{I.~E.~Tamm Theory Department, P.~N.~Lebedev
Physical Institute, Leninsky prospect 53, Moscow 119991, Russia}

\email{tyutin@lpi.ru}

\author{B.L.~Voronov}
\address{I.~E.~Tamm Theory Department, P.~N.~Lebedev
Physical Institute, Leninsky prospect 53, Moscow 119991, Russia}

\email{voronov@lpi.ru}

\thanks{This research was supported by the Russian Foundation for Basic Research (Grant Nos. 09-01-00835, A.G.S.; 08-02-01118, I.V.T.); D.M.G.
is grateful to the Brazilian foundations FAPESP and CNPq for permanent support}

\begin{abstract}
We develop a general technique for finding self-adjoint extensions of a symmetric operator that respect a given set of its symmetries. Problems of
this type naturally arise when considering two- and three-dimensional Schr\"odinger operators with singular potentials. The approach is based on
constructing a unitary transformation diagonalizing the symmetries and reducing the initial operator to the direct integral of a suitable family of
partial operators. We prove that symmetry preserving self-adjoint extensions of the initial operator are in a one-to-one correspondence with
measurable families of self-adjoint extensions of partial operators obtained by reduction. The general construction is applied to the
three-dimensional Aharonov-Bohm Hamiltonian describing the electron in the magnetic field of an infinitely thin solenoid.
\end{abstract}

\maketitle

\section{Introduction}

It is well known that strong singularities in the potential may lead to the lack of self-adjointness of the corresponding Schr\"odinger operator on
its natural domain. As a result, the quantum model is no longer fixed uniquely by the potential and different quantum dynamics described by various
self-adjoint extensions of the initial Schr\"odinger operator are possible. Without additional physical information, it is generally impossible to
choose a single extension giving the ``true'' dynamics. However, the arbitrariness can be reduced if there are symmetries of the initial
Schr\"odinger operator: in this case, it is natural to require the extensions to also respect these symmetries. In this paper, we propose a general
technique for finding all such symmetry preserving extensions and apply it to the analysis of the Aharonov-Bohm Hamiltonian describing a charged
particle in the magnetic field of an infinitely thin solenoid.

Most generally, the problem of finding symmetry preserving self-adjoint extensions can be posed as follows. Suppose $H$ is a symmetric (not
necessarily closed) operator in a separable Hilbert space $\mathfrak H$ and $\mathcal X$ is a set of symmetries of $\mathfrak H$, i.e., bounded
everywhere defined operators in $\mathfrak H$ commuting\footnote{The commutation of $T\in \mathcal X$ with $H$ means that $T\Psi\in D_H$ and
$TH\Psi=HT\Psi$ for any $\Psi$ belonging to the domain $D_H$ of $H$ (see the beginning of Sec.~\ref{s3}).} with $H$. Then our aim is to find all
self-adjoint extensions $\tilde H$ of $H$ that commute with all elements of $\mathcal X$.

In this paper, we assume that the symmetries are normal pairwise commuting operators. The procedure of finding symmetry preserving self-adjoint
extensions of $H$ falls into three major steps:
\begin{itemize}
\item Diagonalization of symmetries.

\item Reduction of $H$.

\item Finding self-adjoint extensions of the partial operators obtained via reduction of $H$.
\end{itemize}
By a diagonalization of $\mathcal X$, we mean a unitary operator $V\colon \mathfrak H\to \int_S^\oplus \mathfrak S(s)\,d\nu(s)$ such that $VTV^{-1}$
is the operator $\mathcal T_f$ of multiplication by some $\nu$-measurable complex function $f$ in $\int_S^\oplus \mathfrak S(s)\,d\nu(s)$ for any
$T\in \mathcal X$ (here, $\nu$ is a measure on a measurable space $S$ and $\mathfrak S$ is a $\nu$-measurable family of Hilbert spaces on $S$; we
shall briefly recall the notions related to direct integrals of Hilbert spaces in Sec.~\ref{s5}). We shall be mainly interested in a special class of
diagonalizations, called exact, that satisfy the following condition:
\begin{itemize}
\item[(E)] For any bounded everywhere defined operator $R$ in $\mathfrak H$ that commutes with all elements of $\mathcal X$, the operator $VRV^{-1}$
commutes with $\mathcal T_f$ for any $\nu$-measurable bounded $f$ on $S$.
\end{itemize}
This condition allows us to apply the von~Neumann's reduction theory~\cite{Neumann} (or, more precisely, its generalization due to
Nussbaum~\cite{Nussbaum} for the case of unbounded operators) and conclude that $VRV^{-1}$ can be decomposed into a direct integral of closed
operators for any closed $R$ commuting with symmetries. For applications, it is important to have a criterion for deciding whether a given
diagonalization is exact or not. To this end, we introduce the notion of a $\nu$-separating family of functions on $S$ (see Definition~\ref{d_exact})
and prove, under very mild assumptions on $\nu$, that a diagonalization is exact if and only if there is a $\nu$-separating family
$\{f_\iota\}_{\iota\in I}$ of $\nu$-measurable complex functions on $S$ such that $V^{-1}\mathcal T_{f_\iota}V\in \mathcal X$ for any $\iota\in I$
(see Theorem~\ref{t0a}). The latter condition is usually easily checked for concrete examples.

By a reduction of $H$ with respect to a given diagonalization for $\mathcal X$, we mean a $\nu$-measurable family of operators $a(s)$ acting in
$\mathfrak S(s)$ such that $\int_S^\oplus a(s)\,d\nu(s)$ is an extension of $VHV^{-1}$ and the image $V(D_H)$ of $D_H$ under $V$ has a suitable
density with respect to the domains of $a(s)$ (see Definition~\ref{d5} and Definition~\ref{d_red} for details).

In this paper, we do not give any general recipe for constructing diagonalizations and reductions: this has to be done separately for each concrete
case. At the same time, we prove that exact diagonalizations and reductions always exist for any set $\mathcal X$ of normal bounded pairwise
commuting operators in $\mathfrak H$ and any densely defined closable operator $H$ commuting with all elements of $\mathcal X$ (Theorem~\ref{t_diag}
and Lemma~\ref{l_red}).

Given an exact diagonalization for $\mathcal X$ and a reduction of $H$, we can describe all symmetry preserving extensions of $H$. Namely, we prove
(Theorem~\ref{t1}) that the operator
\begin{equation}\label{sep8aa}
\tilde H = V^{-1}\int^\oplus \tilde a(s) \,d\nu(s)\, V
\end{equation}
is a self-adjoint extension of $H$ commuting with symmetries for any $\nu$-measurable family $\tilde a(s)$ of self-adjoint extensions of $a(s)$.
Conversely, for any self-adjoint extension $\tilde H$ of $H$ commuting with symmetries, there is a unique (up to $\nu$-equivalence) $\nu$-measurable
family $\tilde a(s)$ of self-adjoint extensions of $a(s)$ such that (\ref{sep8aa}) holds.

We illustrate the general construction described above by applying it to the three-dimensional model of an electron in the magnetic field of an
infinitely thin solenoid. In this case, the Hamiltonian is formally given by the differential expression
\begin{equation}\label{diff_expr_a}
\frac{\hbar^2}{2m_{\mathrm e}}\left(i\nabla + \frac{e}{\hbar c}\mathbf A\right)^2,
\end{equation}
where $e$ and $m_{\mathrm e}$ are the electron charge and mass respectively, $c$ is the velocity of light, and the vector potential $\mathbf
A=(A^1,A^2,A^3)$ has the form
\[
A^1(x,y,z) = -\frac{\Phi y}{2\pi(x^2+y^2)},\quad A^2(x,y,z) = \frac{\Phi x}{2\pi(x^2+y^2)},\quad A^3(x,y,z)=0
\]
($\Phi$ is the flux of the magnetic field through the solenoid). Expression~(\ref{diff_expr_a}) is singular on the $z$-axis. For this reason,
(\ref{diff_expr_a}) naturally determines an operator $H$ in $L^2(\R^3)$ with the domain consisting of smooth functions with compact support separated
from the $z$-axis. As the set $\mathcal X$ of symmetries, it is natural to choose the set of all operators in $L^2(\R^3)$ induced by translations
along the $z$-axis and rotations around the $z$-axis (it is straightforward to check that $H$ commutes with all such operators). We describe all
self-adjoint extensions of $H$ commuting with the elements of $\mathcal X$ (Theorem~\ref{t2}).

The paper is organized as follows. In Sec.~\ref{s2}, we fix the measure-theoretic notation and recall some basic facts concerning the integration
with respect to spectral measures. In Sec.~\ref{s3}, we show how the commutation properties of (unbounded) operators in a Hilbert space can be
described in terms of von~Neumann algebras, which provide a convenient setting for the study of diagonalizations and their exactness. In
Sec.~\ref{s4}, we give the definition of $\nu$-separating families of functions and use it to describe the systems of generators of von~Neumann
algebras associated with spectral measures. In Sec.~\ref{s5}, we reformulate the definition of exact diagonalization in terms of von~Neumann
algebras, establish the existence of exact diagonalizations, and use the results of Secs.~\ref{s3} and~\ref{s4} to characterize them in terms of
$\nu$-separating families. In Sec.~\ref{s6}, we prove the existence of reductions for any symmetric operator with respect to exact diagonalizations
of symmetries and obtain the description of its symmetry preserving self-adjoint extensions. Secs.~\ref{s_meas} and~\ref{s8} are devoted to
application of the abstract construction to Schr\"odinger operators. In Sec~\ref{s_meas}, we derive a condition for the measurability of families of
one-dimensional Schr\"odinger operators and their self-adjoint extensions. Combining this condition with the general results of Sec.~\ref{s6}, we
find all symmetry preserving self-adjoint extensions of the Aharonov-Bohm Hamiltonian determined by~(\ref{diff_expr_a}).

\section{Preliminaries on measures and spectral measures}
\label{s2}

Recall that a set $S$ is called a measurable space if it is equipped with a $\sigma$-algebra $\Sigma_S$ of subsets of $S$. Given a Borel space $S$,
the elements of $\Sigma_S$ are called measurable subsets of $S$. Every subset $A$ of a measurable space $S$ has a natural structure of a measurable
space: the $\sigma$-algebra $\Sigma_A$ consists of all sets of the form $A\cap B$, where $B\in \Sigma_S$. A map $f$ from a measurable space $S$ to a
measurable space $S'$ is called measurable if $f^{-1}(A)\in \Sigma_S$ for any measurable subset $A$ of $S'$. A measurable map $f\colon S\to S'$ is
called a measurable isomorphism if it is bijective and $f^{-1}$ is a measurable map from $S'$ to~$S$.

If $S$ is a topological space, then it can be naturally made a measurable space by putting $\Sigma_S$ equal to the Borel $\sigma$-algebra of $S$
(i.e., the smallest $\sigma$-algebra on $S$ containing all open subsets of $S$). We shall assume, unless otherwise specified, that all considered
topological spaces (in particular, $\R$ and $\C$) carry a measurable structure defined in this way.

A measure on a measurable space $S$ means a countably additive function $\nu$ from $\Sigma_S$ to the extended positive semi-axis $[0,\infty]$. A
subset $N$ of $S$ is called a $\nu$-null set if $N\subset N'$, where $N'$ is measurable and $\nu(N')=0$. A map $f$ is said to be defined $\nu$-almost
everywhere ($\nu$-a.e.) on $S$ if there is a $\nu$-null set $N$ such that $S\setminus N\subset D_f$, where $D_f$ is the domain of $f$. Given a set
$S'$, a map $f$ is said to be an $\nu$-a.e. defined map from $S$ to $S'$ if there is a $\nu$-null set $N$ such that $S\setminus N\subset D_f$ and
$f(s)\in S'$ for any $s\in S\setminus N$. Two $\nu$-a.e. defined maps $f$ and $g$ are called equal $\nu$-a.e. if there is a $\nu$-null set $N$ such
that $S\setminus N\subset D_f\cap D_g$ and $f$ and $g$ coincide on $S\setminus N$. The $\nu$-essential supremum of a $\nu$-a.e. defined real function
$f$ on $S$ (notation $\nu$-$\mathrm{ess\,sup}_{s\in S}f(s)$) is the greatest lower bound of $C\in \R$ such that $f(s)\leq C$ for $\nu$-almost every
($\nu$-a.e.) $s\in S$. A complex $\nu$-a.e. defined function $f$ is said to be $\nu$-essentially bounded on $S$ if $\nu$-$\mathrm{ess\,sup}_{s\in
S}|f(s)|<\infty$. A map $f$ is called an $\nu$-measurable map from $S$ to a measurable space $S'$ if $f$ is defined $\nu$-a.e. on $S$ and there is a
measurable map from $S$ to $S'$ that is $\nu$-a.e. equal to $f$.

All maps defined $\nu$-a.e. on $S$ fall into equivalence classes of maps that are equal $\nu$-a.e. Given a $\nu$-a.e. defined map $f$ on $S$, we
denote its equivalence class by $[f]_\nu$. For any set $S'$, we denote by $\mathcal F(S,S',\nu)$ the set of all equivalence classes $[f]_\nu$ such
that $\xi(s)\in S'$ for $\nu$-a.e. $s\in S$. If $S'$ is a vector space, then $\mathcal F(S,S',\nu)$ obviously has a natural structure of a vector
space. Given a Hilbert space $\mathfrak H$, we denote by $L^2(S,\mathfrak H,\nu)$ the subspace of $\mathcal F(S,S',\nu)$ consisting of all $[f]_\nu$
such that $f$ is a $\nu$-measurable map from $S$ to $\mathfrak H$ and $\int \|f(s)\|^2\,d\nu(s)<\infty$. For $\mathfrak H=\C$, the space
$L^2(S,\mathfrak H,\nu)$ will be denoted by $L^2(S,\nu)$. If $\nu$ is the Lebesgue measure, the space $L^2(S,\nu)$ will be denoted by $L^2(S)$.

A measure $\nu$ on $S$ is called $\sigma$-finite if there is a sequence of measurable sets $A_1,A_2,\ldots$ such that $S=\bigcup_{j=1}^\infty A_j$
and $\nu(A_j)<\infty$ for all $j$. Throughout the paper, all measures will be assumed $\sigma$-finite.

Let $S$ be a measurable space, $\mathfrak H$ be a Hilbert space, and $\mathcal P(\mathfrak H)$ be the set of orthogonal projections on $\mathfrak H$.
A map $E\colon \Sigma_S\to \mathcal P(\mathfrak H)$ is called a spectral measure for $(S,\mathfrak H)$ if it is countably additive with respect to
the strong operator topology on $\mathcal P(\mathfrak H)$ and $E(S)$ is the identity operator in $\mathfrak H$. If $E$ is a spectral measure, then
$E(A_1\cap A_2)=E(A_1)E(A_2)$ for any measurable $A_1,A_2\subset S$ (see~\cite{BirmanSolomjak}, Sec.~5.1, Theorem~1). For any $\Psi\in \mathfrak H$,
the finite positive measure $E_\Psi$ on $S$ is defined by setting $E_\Psi(A)=\langle E(A)\Psi,\Psi\rangle$ for any measurable $A$, where $\langle
\cdot,\cdot\rangle$ is the scalar product on $\mathfrak H$. A subset $N$ of $S$ is called an $E$-null set if $N\subset N'$, where $N'$ is measurable
and $E(N')=0$. The concepts of an $E$-a.e. defined function, an $E$-essentially bounded function, and an $E$-measurable function are defined in the
same way as in the case of a positive measure with $\nu$-null sets replaced with $E$-null sets.

Let $E$ be a spectral measure on a measurable space $S$. Given an $E$-measurable complex function $f$ on $S$, the integral $J^E_f$ of $f$ with
respect to $E$ is defined as the unique linear operator in $\mathfrak H$ such that
\begin{align}
&D_{J^E_f}=\left\{\Psi\in \mathfrak H : \int |f(s)|^2\,dE_\Psi(s)<\infty \right\} \label{dom} \\
& \langle\Psi, J^E_f \Psi\rangle = \int f(s)\,dE_\Psi(s),\quad \Psi\in D_{J^E_f}. \nonumber
\end{align}
For any $E$-measurable complex function $f$ on $S$, the operator $J^E_f$ is normal\footnote{Recall that a closed densely defined linear operator $T$
in a Hilbert space $\mathfrak H$ is called normal if the operators $TT^*$ and $T^*T$ have the same domain of definition and coincide thereon (as
usual, $T^*$ denotes the adjoint of $T$).}, and we have $J^E_{\bar f}=J^*_f$, where $\bar f$ is the complex conjugate function of $f$. For any
$E$-measurable $f$ and $g$ on $S$, we have
\begin{equation}\label{sum_prod}
J^E_{fg}=\overline{J^E_fJ^E_g},\quad  J^E_{f+g}=\overline{J^E_f+J^E_g},
\end{equation}
where the bar means closure. The operator $J^E_f$ is everywhere defined and bounded if and only if $f$ is $E$-essentially bounded. In this case, we
have
\begin{equation}\label{normbound}
\|J^E_f\| = \nu\mbox{-}\mathrm{ess\,sup}_{s\in S} |f(s)|.
\end{equation}

For every normal operator $T$, there is a unique spectral measure $\mathcal E_T$ on $\C$ such that $J^{\mathcal E_T}_{\mathrm{id}_{\C}}=T$, where
$\mathrm{id}_{\C}$ is the identity function on $\C$. The operators $\mathcal E_T(A)$, where $A$ is a Borel subset of $\C$, are called the spectral
projections of $T$. If $f$ is an $\mathcal E_T$-measurable complex function on $S$, then the operator $J_f^{\mathcal E_T}$ is also denoted as $f(T)$.

Let $S$ and $S'$ be measurable spaces, $\mathfrak H$ be a Hilbert space, $E$ be a spectral measure for $(S,\mathfrak H)$, and $\varphi\colon S\to S'$
be an $E$-measurable map. We denote by $\varphi_* E$ the push-forward of $E$ under $\varphi$. By definition, this means that $\varphi_* E$ is the
spectral measure for $(S',\mathfrak H)$ such that
\[
\varphi_* E(A) = E(\varphi^{-1}(A))
\]
for any measurable $A\subset S'$. If $f$ is an $\varphi_*E$-measurable complex function on $S'$, then $f\circ \varphi$ is an $E$-measurable function
on $S$, and we have
\begin{equation}\label{pushforward}
J^{\varphi_*E}_f = J^E_{f\circ \varphi}.
\end{equation}
Let $\varphi$ be an $E$-measurable complex function on $S$. Formula~(\ref{pushforward}) with $f=\mathrm{id}_{\C}$ yields $J^{\varphi_*
E}_{\mathrm{id}_{\C}}=J^E_\varphi$. In view of the uniqueness of $\mathcal E_{J^E_\varphi}$, this means that
\begin{equation}\label{600}
\mathcal E_{J^E_\varphi}=\varphi_* E.
\end{equation}
In view of~(\ref{pushforward}), it follows that $f\circ\varphi$ is $E$-measurable and
\begin{equation}\label{composition}
f(J^E_\varphi) = J^{\varphi_* E}_f = J^E_{f\circ \varphi}
\end{equation}
for any any $\mathcal E_{J^E_\varphi}$-measurable complex function $f$ on $\C$.

\section{Commutation of operators and von~Neumann algebras}
\label{s3}

Given a Hilbert space $\mathfrak H$, we denote by $L(\mathfrak H)$ the algebra of all bounded everywhere defined linear operators in $\mathfrak H$.
We say that operators $T_1$ and $T_2$ in $\mathfrak H$ with the respective domains $D_{T_1}$ and $D_{T_2}$ commute if one of the following conditions
is satisfied:
\begin{enumerate}
\item[(A)] $T_1\in L(\mathfrak H)$ and we have $T_1\Psi\in D_{T_2}$ and $T_1T_2\Psi=T_2T_1 \Psi$ for any $\Psi\in D_{T_2}$.

\item[(B)] $T_2\in L(\mathfrak H)$ and we have $T_2\Psi\in D_{T_1}$ and $T_1T_2\Psi=T_2T_1 \Psi$ for any $\Psi\in D_{T_1}$.

\item[(C)] $T_1$ and $T_2$ are both normal and their spectral projections commute.

\end{enumerate}

Clearly, $T_1$ commutes with $T_2$ if and only if $T_2$ commutes with $T_1$. The next lemma shows that, wherever applicable, conditions~(A), (B), and
(C) are equivalent.

\begin{lemma} \label{l0}
Let $T_1$ and $T_2$ be commuting operators in $\mathfrak H$. Then we have
\begin{enumerate}
\item[1.] If $T_1\in L(\mathfrak H)$, then {\rm (A)} is satisfied.

\item[2.] If $T_2\in L(\mathfrak H)$, then {\rm (B)} is satisfied.

\item[3.] If $T_1$ and $T_2$ are both normal, then {\rm (C)} is satisfied.
\end{enumerate}
\end{lemma}

\begin{proof}

1. Let $T_1\in L(\mathfrak H)$. If (B) is satisfied, then $T_2\in L(\mathfrak H)$ and $T_1$ and $T_2$ commute in the usual sense. Hence, (A) holds.
If (C) is satisfied, then (A) is ensured by Theorem~8 of Sec.~5.4 in~\cite{BirmanSolomjak} applied to the spectral measure of $T_2$.

\par\noindent
2. Let $T_2\in L(\mathfrak H)$. By statement~1, $T_2$ and $T_1$ satisfy (A), i.e., $T_1$ and $T_2$ satisfy~(B).

\par\noindent
3. Let $T_1$ and $T_2$ be both normal. If (A) is satisfied, then $T_1$ commutes with all spectral projections of $T_2$ by Theorem~1
in~\cite{Fuglede}. Applying the latter theorem again, we get~(C). If (B) holds, then $T_2$ and $T_1$ satisfy~(A), which again implies~(C). The lemma
is proved.
\end{proof}

\begin{lemma} \label{l1}
Let $R\in L(\mathfrak H)$ and $T$ be a densely defined operator in $\mathfrak H$ commuting with $R$. Then $R^*$ commutes with $T^*$. If $T$ is
closable, then the closure $\bar T$ of $T$ commutes with $R$.
\end{lemma}
\begin{proof}
By statement~1 of Lemma~\ref{l0}, the operators $R$ and $T$ satisfy~(A). Let $\Psi\in D_{T^*}$ and $\Phi=T^*\Psi$. Then we have $\langle
T\Psi',\Psi\rangle=\langle \Psi',\Phi\rangle$ for any $\Psi'\in D_T$ and, in view of~(A), we obtain
\[
\langle T\Psi',R^*\Psi\rangle=\langle TR\Psi',\Psi \rangle=\langle R\Psi',\Phi \rangle = \langle \Psi',R^*\Phi \rangle,\quad \Psi'\in D_T.
\]
This means that $R^*\Psi\in D_{T^*}$ and $T^* R^* \Psi=R^*T^*\Psi$, i.e., $R^*$ and $T^*$ satisfy~(A). If $T$ is closable, then $R$ commutes with
$\bar T$ because $R=(R^*)^*$ and $\bar T = (T^*)^*$. The lemma is proved.
\end{proof}

Given a set $\mathcal X$ of operators in $\mathfrak H$, let $\mathcal X'$ denote its commutant, i.e., the subalgebra of $L(\mathfrak H)$ consisting
of all operators commuting with every element of $\mathcal X$. If all operators in $\mathcal X$ are densely defined, we denote by $\mathcal X^*$ the
set consisting of the adjoints of the elements of $\mathcal X$. The set $\mathcal X$ is called involutive if $\mathcal X^*=\mathcal X$.
Lemma~\ref{l1} implies that
\begin{equation}\label{adjcomm}
(\mathcal X')^*=(\mathcal X^*)'
\end{equation}
whenever all elements of $\mathcal X$ are closed and densely defined. Recall~\cite{Dixmier} that a subalgebra $\mathcal A$ of $L(\mathfrak H)$ is
called a von~Neumann algebra if it is involutive and coincides with its bicommutant~$\mathcal A''$. By the well-known von~Neumann's theorem (see,
e.~g.,~\cite{Dixmier}, Sec.~I.3.4, Corollaire~2), an involutive subalgebra $\mathcal A$ of $L(\mathfrak H)$ is a von~Neumann algebra if and only if
it contains the identity operator and is closed in the strong\footnote{The same is true for the weak operator topology because every involutive
strongly closed subalgebra of $L(\mathfrak H)$ is weakly closed (see~\cite{Dixmier}, Sec.~I.3.4, Th\'eor\`eme~2).} operator topology.

\begin{lemma} \label{l2}
If $\mathcal X$ is an involutive set of densely defined closed operators in $\mathfrak H$, then $\mathcal X'$ is a von~Neumann algebra.
\end{lemma}
\begin{proof}
By~(\ref{adjcomm}), the algebra $\mathcal X'$ is involutive, and it suffices to show that $\mathcal X'$ is closed in the strong operator topology.
Given an operator $T$ in $\mathfrak H$, let $C_T$ denote the set of all elements of $L(\mathfrak H)$ commuting with $T$ (in other words, $C_T$ is the
commutant of the one-point set $\{T\}$). Since $\mathcal X'=\bigcap_{T\in\mathcal X} C_T$, it suffices to prove that $C_T$ is strongly closed for any
closed $T$. Let $R$ belong to the strong closure of $C_T$. For every $\Psi_1,\Psi_2\in\mathfrak H$ and $n=1,2,\ldots$, the set
\[
W_{\Psi_1,\Psi_2,n}=\{\tilde R\in L(\mathfrak H) : \|(\tilde R-R)\Psi_i\|<1/n,\,\,i=1,2\}
\]
is a strong neighborhood of $R$ and, hence, has a nonempty intersection with $C_T$. Fix $\Psi\in D_T$ and choose $R_n\in C_T\cap
W_{\Psi,\,T\Psi,\,n}$ for each $n$. Then $R_n \Psi\to R\Psi$ and $R_n T \Psi\to RT\Psi$ in $\mathfrak H$. As $R_n$ commute with $T$, we have $R_n
\Psi\in D_T$ and $R_nT \Psi= TR_n \Psi$ for all $n$. In view of the closedness of $T$, it follows that $R\Psi\in D_T$ and $TR\Psi = RT\Psi$, i.e.,
$R\in C_T$. The lemma is proved.
\end{proof}

Let $\mathcal X$ be a set of closed densely defined operators in $\mathfrak H$. Then the set $\mathcal X\cup \mathcal X^*$ is involutive. We set
$\mathcal A(\mathcal X)=(\mathcal X\cup \mathcal X^*)''$ and call $\mathcal A(\mathcal X)$ the von~Neumann algebra generated by $\mathcal X$. If
$\mathcal X\subset L(\mathfrak H)$, then $\mathcal A(\mathcal X)$ is the smallest von~Neumann algebra containing $\mathcal X$. If $\mathcal X$
consists of normal operators, then $(\mathcal X\cup \mathcal X^*)'=\mathcal X'$ (by Theorem~1 in~\cite{Fuglede}, if $R\in L(\mathfrak H)$ commutes
with a normal operator $T$, then it commutes with $T^*$) and, therefore, we have $\mathcal A(\mathcal X)=\mathcal X''$. If $T$ is a closed densely
defined operator, then we shall write $\mathcal A(T)$ instead of $\mathcal A(\{T\})$, where $\{T\}$ is the one-point set containing~$T$.

Given a spectral measure $E$ on a measurable space $S$, we denote by $\mathcal P_E$ the set of all operators $E(A)$, where $A$ is a measurable subset
of $S$. Theorem~3 of Sec.~6.6 in~\cite{BirmanSolomjak} implies that $\{T,T^*\}'=\mathcal P'_{\mathcal E_T}$ for any normal operator $T$ and, hence,
\begin{equation}\label{601}
\mathcal A(T)=\mathcal A(\mathcal P_{\mathcal E_T}).
\end{equation}

\begin{lemma}\label{lll}
Let $T_1$ and $T_2$ be normal operators in $\mathfrak H$. Then the following statements are equivalent:
\begin{enumerate}
\item $T_1$ commutes with $T_2$.

\item $T_1$ commutes with every element of $\mathcal A(T_2)$.

\item Every element of $\mathcal A(T_1)$ commutes with every element of $\mathcal A(T_2)$.
\end{enumerate}
\end{lemma}
\begin{proof}
Let $\mathcal P_1$ and $\mathcal P_2$ be the sets of all spectral projections of $T_1$ and $T_2$ respectively. By~(\ref{601}), we have
\begin{equation}\label{at1} \mathcal
A(T_1)=\mathcal A(\mathcal P_1),\quad \mathcal A(T_2)=\mathcal A(\mathcal P_2).
\end{equation}
By Lemma~\ref{l0}, $T_1$ and $T_2$ commute if and only if they satisfy~(C), i.e., if and only if $\mathcal P_1\subset \mathcal P'_2$. Since $\mathcal
P'_2 = \mathcal A(\mathcal P_2)'$ and $\mathcal A(\mathcal P_1)$ is the smallest von~Neumann algebra containing $\mathcal P_1$, the latter inclusion
is equivalent to the relation $\mathcal A(\mathcal P_1)\subset \mathcal A(\mathcal P_2)'$, which means, in view of~(\ref{at1}), that all elements of
$\mathcal A(T_1)$ commute with all elements of $\mathcal A(T_2)$. Thus, statements~1 and~3 are equivalent. By Lemma~\ref{l1}, $T_1$ commutes with
every element of $\mathcal A(T_2)$ if and only if $\mathcal A(T_2)\subset \{T_1,T_1^*\}'$, i.e., if and only if $\mathcal A(T_1)\subset \mathcal
A(T_2)'$. Hence, statements~2 and~3 are equivalent and the lemma is proved.
\end{proof}

\begin{lemma}\label{l_norm}
A closed densely defined operator $T$ in $\mathfrak H$ is normal if and only if the algebra $\mathcal A(T)$ is Abelian.
\end{lemma}
\begin{proof}
If $T$ is normal, then it commutes with itself and Lemma~\ref{lll} shows that $\mathcal A(T)$ is Abelian. Let $T$ be a closed densely defined
operator such that $\mathcal A(T)$ is Abelian. Let $|T|=(T^*T)^{1/2}$. Then $D_{|T|}=D_T$, the range $\mathrm{Ran}\,|T|$ of $|T|$ coincides with
$\mathrm{Ran}\,T^*$, and we have the polar decompositions
\begin{equation}\label{polar_dec}
T=U|T|,\quad T^* = |T|U^*
\end{equation}
where $U$ is a partially isometric operator $U$ in $\mathfrak H$ with the initial space $\overline{\mathrm{Ran}\,|T|}$ (see, e.g., Sec.~8.1
in~\cite{BirmanSolomjak} for details on polar decomposition). By Theorem~4 of Sec.~8.1 in~\cite{BirmanSolomjak}, we have
\begin{equation}\label{tt*}
TT^* = U T^*T U^*.
\end{equation}
We now show that
\begin{equation}\label{uu}
U\in \mathcal A(T), \quad \mathcal A(|T|)\subset \mathcal A(T).
\end{equation}
Let $R\in \{T,T^*\}'$. Then $R$ commutes with $T^*T$, and it follows from Theorem~2 of Sec.~6.3 and Theorem~8 of Sec.~5.4 in~\cite{BirmanSolomjak}
that $R$ commutes with $|T|$. This implies the second relation in~(\ref{uu}). Let $\Psi \in \mathrm{Ran}\,|T|$ and $\Phi\in D_T$ be such that
$\Psi=|T|\Phi$. Then $R\Phi\in D_T$, and we have
\[
RU\Psi = R T\Phi = T R\Phi = U |T| R\Phi = UR\Psi.
\]
If $\Psi\in \mathrm{Ran}\,|T|^\bot$, then $R\Psi\in \mathrm{Ran}\,|T|^\bot$ because $R$ and $|T|$ commute, and we have $UR\Psi = RU\Psi=0$. We thus
see that $RU=UR$ everywhere on $\mathfrak H$ and, therefore, the first relation in~(\ref{uu}) holds. Hence, $U$ is normal, and it follows
from~(\ref{uu}) and Lemma~\ref{lll} that $U$ commutes with $|T|$. Now equalities~(\ref{polar_dec}) imply that $U$ commutes with both $T$ and $T^*$,
and Lemma~\ref{l1} ensures that $U^*$ also commutes with both $T$ and $T^*$. Hence, $U^*$ commutes with $UT^*T$, and it follows from~(\ref{tt*}) that
$TT^*$ is an extension of the operator $U^*UT^*T$. But $U^*U$ is the orthogonal projection onto the initial space
$\overline{\mathrm{Ran}\,|T|}=\overline{\mathrm{Ran}\,T^*}$ and, therefore, $U^*UT^*T=T^*T$. Thus, $TT^*$ is an extension of $T^*T$. Since both
operators are self-adjoint, this implies $TT^*=T^*T$. The lemma is proved.
\end{proof}

If all elements of an involutive set $\mathcal X\subset L(\mathfrak H)$ pairwise commute, then the algebra $\mathcal A(\mathcal X)$ is Abelian.
Indeed, we have $\mathcal X\subset \mathcal X'$ and, therefore, $\mathcal A(\mathcal X)\subset \mathcal X'$, whence the statement follows because
$\mathcal X'=\mathcal A(\mathcal X)'$.

\begin{lemma}\label{l4d}
Let $S$ be a measurable space, $\mathfrak H$ be a separable Hilbert space, and $E$ be a spectral measure for $(\mathfrak H,S)$. Then $\mathcal
A(\mathcal P_E)$ coincides with the set of all $J^E_f$, where $f$ is an $E$-measurable $E$-essentially bounded complex function on $S$. A closed
densely defined operator $T$ in $\mathfrak H$ is equal to $J^E_f$ for an $E$-measurable complex function $f$ on $S$ if and only if
\begin{equation}\label{at}
\mathcal A(T)\subset \mathcal A(\mathcal P_E)
\end{equation}
\end{lemma}
\begin{proof}
By Theorem~8 of Sec.~5.4 in~\cite{BirmanSolomjak}, we have $\mathcal A(J^E_f)\subset \mathcal A(\mathcal P_E)$ for any $E$-measurable complex
function $f$ on $S$. This implies that $J^E_f\in \mathcal A(\mathcal P_E)$ for $E$-essentially bounded $f$ because $J^E_f$ belongs to $L(\mathfrak
H)$ for such $f$ and, hence, is contained in $\mathcal A(J^E_f)$. Conversely, Theorem~5 of Sec.~7.4 in~\cite{BirmanSolomjak} shows that any element
of $\mathcal A(\mathcal P_E)$ is equal to $J^E_f$ for some $E$-measurable $E$-essentially bounded complex function $f$ on $S$. It remains to prove
that any closed densely defined operator $T$ such that (\ref{at}) holds is equal to $J^E_f$ for some $E$-measurable complex function $f$ on $S$.
Since the elements of $\mathcal P_E$ pairwise commute, the algebra $\mathcal A(\mathcal P_E)$ is Abelian, and Lemma~\ref{l_norm} implies that $T$ is
normal. Let $\chi$ be a complex function on $\C$ defined by the relation
\[
\chi(z) = \frac{z}{|z|+1}.
\]
It is easy to see that the function $\chi$ is one-to-one and maps $\C$ onto the open unit disc $\mathbb D = \{z\in\C : |z|<1\}$. Its inverse function
$\chi^{-1}$ is given by
\[
\chi^{-1}(z) = \frac{z}{1-|z|},\quad |z|<1.
\]
Since $\chi$ is bounded and measurable on $\C$, we have $\chi(T)\in\mathcal A(\mathcal P_{\mathcal E_T})$. By~(\ref{601}), it follows that
$\chi(T)\in\mathcal A(T)$. In view of~(\ref{at}), this implies that $\chi(T)\in \mathcal A(\mathcal P_E)$ and, therefore, $\chi(T)=J^E_g$ for some
$E$-measurable function $g$ on $S$. As $\C\setminus \mathbb D$ is a $\chi_*\mathcal E_T$-null set and $\chi^{-1}$ is a measurable map from $\mathbb
D$ to $\C$, the function $\chi^{-1}$ is $\chi_*\mathcal E_T$-measurable on $\C$. By~(\ref{600}), we have $\mathcal E_{\chi(T)}=\chi_* \mathcal E_T$
and, hence, $\chi^{-1}$ is $\mathcal E_{\chi(T)}$-measurable. It therefore follows from~(\ref{composition}) that
\[
T = \chi^{-1}(\chi(T)) = \chi^{-1}(J^E_g) = J^E_f
\]
for $f=\chi^{-1}\circ g$. The lemma is proved.
\end{proof}

We say that two sets $\mathcal X$ and $\mathcal Y$ of closed densely defined operators in $\mathfrak H$ are equivalent if $\mathcal A(\mathcal X) =
\mathcal A(\mathcal Y)$. We say that $\mathcal X$ is equivalent to a closed densely defined operator $T$ if $\mathcal X$ is equivalent to the
one-point set $\{T\}$. Two closed densely defined operators $T_1$ and $T_2$ are called equivalent if $\{T_1\}$ and $\{T_2\}$ are equivalent.

\begin{lemma}\label{l4}
Let $\{\mathcal X_\iota\}_{\iota\in I}$ and $\{\mathcal Y_\iota\}_{\iota\in I}$ be families of sets of closed densely defined operators in $\mathfrak
H$ and let $\mathcal X=\bigcup_{\iota\in I}\mathcal X_\iota$ and $\mathcal Y=\bigcup_{\iota\in I}\mathcal Y_\iota$. If $\mathcal X_\iota$ and
$\mathcal Y_\iota$ are equivalent for every $\iota\in I$, then $\mathcal X$ and $\mathcal Y$ are equivalent.
\end{lemma}
\begin{proof}
Set $M_\iota = (\mathcal X_\iota\cup \mathcal X^*_\iota)'$ and $M = (\mathcal X\cup \mathcal X^*)'$. Then we have $M=\bigcap_{\iota\in I} M_\iota$.
Hence $M'=\mathcal A(\mathcal X)$ coincides with the von~Neumann algebra generated by $\bigcup_{\iota\in I}M'_\iota=\bigcup_{\iota\in I}\mathcal
A(\mathcal X_\iota)$ (see~\cite{Dixmier}, Sec.~I.1.1, Proposition~1). Analogously, $\mathcal A(\mathcal Y)$ is the von~Neumann algebra generated by
$\bigcup_{\iota\in I}\mathcal A(\mathcal Y_\iota)$. Since $\mathcal A(\mathcal X_\iota)=\mathcal A(\mathcal Y_\iota)$ for all $\iota$, it follows
that $\mathcal A(\mathcal X)=\mathcal A(\mathcal Y)$. The lemma is proved.
\end{proof}

\begin{lemma}\label{l3a}
Let $\mathcal X$ be a set of closed densely defined operators in a Hilbert space $\mathfrak H$ and let $T$ be a closed densely defined operator such
that both $T$ and $T^*$ commute with all elements of $\mathcal X$. Then $T$ commutes with all elements of $\mathcal A(\mathcal X)$.
\end{lemma}
\begin{proof}
We first show that $T$ commutes with all elements of $\mathcal A(R)$ for any $R\in \mathcal X$. If $T\in L(\mathfrak H)$, then Lemma~\ref{l1} implies
that $R^*$ commutes with $T$. This means that $T\in \{R,R^*\}'$ and, therefore, $T$ commutes with all elements of $\mathcal A(R) = \{R,R^*\}''$. If
$R\in L(\mathfrak H)$, then $R\in \{T,T^*\}'$. By Lemma~\ref{l2}, $\{T,T^*\}'$ is a von~Neumann algebra. Since $\mathcal A(R)$ is the smallest
von~Neumann algebra containing $R$, we have $\mathcal A(R)\subset \{T,T^*\}'$ and, hence, $T$ commutes with all elements of $\mathcal A(R)$. If
neither $T$ nor $R$ belongs to $L(\mathfrak H)$, then $T$ and $R$ are normal and the statement follows from Lemma~\ref{lll}. Interchanging the roles
of $T$ and $T^*$, we conclude that $T^*$ also commutes with all elements of $\mathcal A(R)$ for any $R\in \mathcal X$. Let $\mathcal Y =
\bigcup_{R\in \mathcal X}\mathcal A(R)$. Clearly, we have $\mathcal Y\subset \{T,T^*\}'$, and Lemma~\ref{l4} implies that $\mathcal A(\mathcal Y)=
\mathcal A(\mathcal X)$. Since $\mathcal A(\mathcal X)$ is the smallest von~Neumann algebra containing $\mathcal Y$, we have $\mathcal A(\mathcal
X)\subset \{T,T^*\}'$ and, hence, $T$ commutes with all elements of $\mathcal A(\mathcal X)$. The lemma is proved.
\end{proof}

\begin{lemma}\label{l41}
Let $\mathfrak H$ be a separable Hilbert space and $\mathcal X$ be a set of closed densely defined operators in $\mathfrak H$. Then there is a
countable subset $\mathcal X_0$ of $\mathcal X$ which is equivalent to $\mathcal X$.
\end{lemma}
\begin{proof}
We first note that every subset of $L(\mathfrak H)$ is separable in the strong topology. Indeed, for any $M\subset L(\mathfrak H)$, we have $M =
\bigcup_{n=1}^\infty M\cap\mathcal B_n$, where $\mathcal B_n=\{T\in L(\mathfrak H) : \|T\|\leq n\}$ is the ball of radius $n$ in $L(\mathfrak H)$.
Since $\mathfrak H$ is separable, $\mathcal B_n$ endowed with the strong topology is a separable metrizable space for any $n$
(see,~e.g.,~\cite{Dixmier}, Sec.~I.3.1). This implies that $M\cap\mathcal B_n$ is separable for any $n$ and, hence, $M$ is separable in the strong
topology.

Let $\mathfrak A = \bigcup_{\mathcal Y\subset \mathcal X} \mathcal A(\mathcal Y)$, where $\mathcal Y$ runs through all finite subsets of $\mathcal
X$. Obviously, $\mathcal A(T)$ is equivalent to $T$ for any closed densely defined operator $T$ and, therefore, Lemma~\ref{l4} implies that $\mathcal
X$ is equivalent to $\bigcup_{T\in \mathcal X}\mathcal A(T)$. Since the latter set is contained in $\mathfrak A$ and $\mathfrak A\subset \mathcal
A(\mathcal X)$, we conclude that $\mathfrak A$ is equivalent to $\mathcal X$. We now note that $\mathfrak A$ is an involutive subalgebra of
$L(\mathfrak H)$ containing the identity operator and, therefore, is strongly dense in $\mathfrak A''=\mathcal A(\mathcal X)$ (\cite{Dixmier},
Sec.~I.3.4, Lemma~6). Let $\mathfrak R$ be a strongly dense countable subset of $\mathfrak A$. For any $R\in\mathfrak R$, we choose a finite set
$\mathcal Y_R\subset \mathcal X$ such that $R\in \mathcal A(\mathcal Y_R)$ and put $\mathcal X_0=\bigcup _{R\in\mathfrak R}\mathcal Y_R$. Clearly,
$\mathcal X_0$ is a countable set. The algebra $\mathcal A(\mathcal X_0)$ is strongly dense in $\mathcal A(\mathcal X)$ because it contains
$\mathfrak R$. On the other hand, $\mathcal A(\mathcal X_0)$ is a von~Neumann algebra and, therefore, is strongly closed. We hence have $\mathcal
A(\mathcal X_0)=\mathcal A(\mathcal X)$, i.e., $\mathcal X_0$ is equivalent to $\mathcal X$. The lemma is proved.
\end{proof}

\begin{lemma}\label{l_com}
Let $\mathfrak H$ be a Hilbert space and $\mathcal X$ be a set of closed densely defined operators in $\mathfrak H$. The algebra $\mathcal A(\mathcal
X)$ is Abelian if and only if all elements of $\mathcal X$ are normal and pairwise commute.
\end{lemma}
\begin{proof}
Let $\mathcal Y = \bigcup_{T\in \mathcal X} \mathcal A(T)$. By Lemma~\ref{l4}, the sets $\mathcal X$ and $\mathcal Y$ are equivalent. If all elements
of $\mathcal X$ are normal and pairwise commute, then Lemma~\ref{lll} implies that all elements of $\mathcal Y$ pairwise commute. Since $\mathcal
Y\subset L(\mathfrak H)$, this means that $\mathcal A(\mathcal Y)=\mathcal A(\mathcal X)$ is Abelian. Conversely, if $\mathcal A(\mathcal X)$ is
Abelian, then all elements of $\mathcal Y$ pairwise commute. Hence, Lemma~\ref{l_norm} implies that all elements of $\mathcal X$ are normal and
Lemma~\ref{lll} implies that all elements of $\mathcal X$ pairwise commute. The lemma is proved.
\end{proof}

\section{Generators of von~Neumann algebras associated \\ with spectral measures}
\label{s4}

Recall that a topological space $S$ is called a Polish space if its topology can be induced by a metric that makes $S$ a separable complete space. A
measurable space $S$ is called a standard Borel space if its measurable structure can be induced by a Polish topology on $S$. A measure $\nu$ on a
measurable space $S$ is called standard if there is a measurable set $S'\subset S$ such that $\nu(S\setminus S')=0$ and $S'$, considered as a
measurable subspace of $S$, is a standard Borel space. Standard spectral measures are defined in the same way.

A family of maps $\{f_\iota\}_{\iota\in I}$ is said to separate points of a set $S$ if for any two distinct elements $s_1$ and $s_2$ of $S$, there is
$\iota\in I$ such that $f_\iota(s_1)\neq f_\iota(s_2)$.

\begin{definition}\label{d_exact}
Let $S$ be a measurable space and $\nu$ be a positive measure on $S$. A family $\{f_\iota\}_{\iota\in I}$ of maps is said to be $\nu$-separating on
$S$ if $I$ is countable and $\{f_\iota\}_{\iota\in I}$ separates points of $S\setminus N$ for some $\nu$-null set $N$. The notion of an
$E$-separating family for a spectral measure $E$ is defined analogously.
\end{definition}

Given a spectral measure $E$ on $\mathfrak H$, we denote by $\mathcal P_E$ the set of all operators $E(A)$, where $A$ is a measurable set. The main
result of this section is the next theorem that gives a complete description of systems of generators for $\mathcal A(\mathcal P_E)$.

\begin{theorem}\label{t0}
Let $S$ be a measurable space, $\mathfrak H$ be a separable Hilbert space, $E$ be a standard spectral measure for $(S,\mathfrak H)$, and $\mathcal X$
be a set of closed densely defined operators in $\mathfrak H$. Then $\mathcal A(\mathcal X)=\mathcal A(\mathcal P_E)$ if and only if the following
conditions hold
\begin{enumerate}
\item For every $T\in \mathcal X$, there is an $E$-measurable complex function $f$ on $S$ such that $T=J^E_f$.

\item There is an $E$-separating family $\{f_\iota\}_{\iota\in I}$ of $E$-measurable complex functions on $S$ such that $J^E_{f_\iota}\in\mathcal X$
for all $\iota\in I$.

\end{enumerate}
\end{theorem}

The rest of this section is devoted to the proof of Theorem~\ref{t0}. Given a topological space $S$, we denote by $C(S)$ the space of all continuous
complex functions on $S$.

\begin{lemma}\label{l01}
Let $S$ be a Polish space, $\mathfrak H$ be a Hilbert space, and $E$ be a spectral measure for $(S,\mathfrak H)$. Let $\mathcal C$ be a subset of
$C(S)$ that separates the points of $S$ and $\mathcal X$ be the set of all operators $J^E_f$ with $f\in\mathcal C$. Then $\mathcal A(\mathcal
X)=\mathcal A(\mathcal P_E)$.
\end{lemma}

In the proof below, all spectral integrals are taken with respect to $E$, and we write for brevity $J_f$ instead of $J^E_f$.
\begin{proof}
Since $\mathcal A(\mathcal X)\subset \mathcal A(\mathcal P_E)$ by Lemma~\ref{l4d}, we have to show that $\mathcal A(\mathcal P_E)\subset \mathcal
A(\mathcal X)$.

Let $\bar{\mathcal C}$ denote the set of functions, complex conjugate to the elements of $\mathcal C$, and let $\mathfrak A$ be the subalgebra of
$C(S)$ generated by $\mathcal C\cup \bar{\mathcal C}$ and the constant functions. Fix $U\in (\mathcal X\cup \mathcal X^*)'$ and let $\mathfrak A_U$
denote the subset of $C(S)$ consisting of all $f$ such that $J_f$ commutes with $U$. If $f,g\in \mathfrak A_U$, then both $J_fJ_g$ and $J_f+J_g$
commute with $U$. In view of Lemma~\ref{l1} and relations~(\ref{sum_prod}), it follows that both $J_{fg}$ and $J_{f+g}$ commute with $U$, i.e.,
$fg\in\mathfrak A_U$ and $f+g\in\mathfrak A_U$. Hence, $\mathfrak A_U$ is an algebra. Since $\mathfrak A_U$ obviously contains $\mathcal C\cup
\bar{\mathcal C}$ and all constant functions, we have $\mathfrak A_U\supset \mathfrak A$. Thus, every element of $(\mathcal X\cup \mathcal X^*)'$
commutes with any operator $J_f$ with $f\in\mathfrak A$.

Given $f\in C(S)$ and a compact set $K\subset S$, we set $B_{f,K}=J_f E(K)$. Let $\Psi\in \mathfrak H$ and $\Phi=E(K)\Psi$. Then for any measurable
set $A$, we have $E_\Phi(A)=E_\Psi(A\cap K)$ and, therefore, $E_\Phi$ is a finite measure supported by $K$. In view of~(\ref{dom}), this implies that
$\Phi\in D_{J_f}$, i.e., the range of $E(K)$ is contained in the domain of $J_f$. Since $E(K)=J_{\chi_K}$, where $\chi_K$ is the characteristic
function of $K$, it follows from~(\ref{sum_prod}) that $B_{f,K}=J_{f\chi_K}$. Hence, $B_{f,K}\in L(\mathfrak H)$ and~(\ref{normbound}) implies that
\begin{equation}\label{norm}
\|B_{f,K}\|\leq \sup_{s\in K}|f(s)|.
\end{equation}
If $f,g\in C(S)$, then~(\ref{sum_prod}) implies that $B_{f+g,K}=B_{f,K}+B_{g,K}$.

We now show that
\begin{equation}\label{incl}
(\mathcal X\cup \mathcal X^*)'\subset {\mathcal Y}',
\end{equation}
where $\mathcal Y$ is the set of all $J_f$ with $f\in C(S)$. Let $U\in (\mathcal X\cup \mathcal X^*)'$. We first prove that $\mathcal Y'$ contains
all operators $U_K=E(K)UE(K)$, where $K$ is a compact subset of $S$. Fix $f\in C(S)$ and let $\varepsilon>0$. Since $\mathcal C\subset \mathfrak A$,
the algebra $\mathfrak A$ separates points of $S$, and the Stone--Weierstrass theorem implies that there is $g\in \mathfrak A$ such that
$|f(s)-g(s)|< \varepsilon$ for any $s\in K$. Since $U_K$ commutes with both $J_g$ and $E(K)$, it follows that $U_K$ commutes with $B_{g,K}$. In view
of~(\ref{norm}), we have
\begin{multline}\nonumber
\|B_{f,K}U_K-U_K B_{f,K}\|\leq \|B_{f,K}U_K-B_{g,K}U_K\|+ \|U_K B_{g,K}-U_K B_{f,K}\|\leq \\ \leq 2\|B_{f-g,K}\|\|U\|<2\varepsilon\|U\|.
\end{multline}
Because $\varepsilon$ is arbitrary, this means that $U_K$ commutes with $B_{f,K}$. This implies that $U_K$ commutes with $J_f$ because
$B_{f,K}U_K=J_f U_K$ and $U_KB_{f,K}$ is an extension of $U_K J_f$ by the commutativity of $E(K)$ and $J_f$. This proves that $U_K\in \mathcal Y'$.
By Lemma~\ref{l2}, $\mathcal Y'$ is a von~Neumann algebra and, in particular, is strongly closed. Hence, inclusion~(\ref{incl}) will be proved if we
demonstrate that every strong neighborhood of $U$ contains $U_K$ for some compact set $K$. To this end, it suffices to show that for every
$\Psi\in\mathfrak H$ and $\varepsilon>0$, there is a compact set $K_{\Psi,\varepsilon}$ such that
\begin{equation}\label{est}
\|(U-U_K)\Psi\|\leq \varepsilon
\end{equation}
for any compact set $K\supset K_{\Psi,\varepsilon}$. Since
\[
U-U_K = E(K)UE(S\setminus K) + E(S\setminus K)U,
\]
we have
\[
\|(U-U_K)\Psi\|\leq \|U\|\, E_\Psi(S\setminus K)^{\frac{1}{2}} + E_\Phi(S\setminus K)^{\frac{1}{2}},
\]
where $\Phi = U\Psi$. As $S$ is a Polish space, Theorem~1.3 in~\cite{Billingsley} ensures that there is a compact set $K_{\Psi,\varepsilon}$ such
that both $E_\Psi(S\setminus K_{\Psi,\varepsilon})$ and $E_\Phi(S\setminus K_{\Psi,\varepsilon})$ do not exceed $\varepsilon^2/4$ and, therefore,
(\ref{est}) holds for any $K\supset K_{\Psi,\varepsilon}$. Inclusion~(\ref{incl}) is thus proved.

We next show that
\begin{equation}\label{incl1}
\mathcal Y'\subset \mathcal P'_E.
\end{equation}
Let $U\in\mathcal Y'$. For any closed set $F\subset S$, it is easy to construct a uniformly bounded sequence of functions $f_n\in C(S)$ that
converges pointwise to $\chi_F$. Then $J_{f_n}$ strongly converge to $J_{\chi_F}=E(F)$ (see Theorem~2 of Sec.~5.3 in~\cite{BirmanSolomjak}). Since
$J_{f_n}$ commute with $U$ for all $n$, this implies that $E(F)$ commutes with $U$. Let $\Sigma^U$ denote the set of all measurable sets $A\subset S$
such that $E(A)$ commutes with $U$. We have proved that $\Sigma^U$ contains all closed sets. If $A\in\Sigma^U$, then $E(S\setminus A)=1-E(A)$
commutes with $U$ and, hence, $S\setminus A\in\Sigma^U$. If $A_1,A_2\in\Sigma^U$, then both $E(A_1\cap A_2)=E(A_1)E(A_2)$ and $E(A_1\cup A_2)=
E(A_1)+E(A_2)-E(A_1)E(A_2)$ commute with $U$ and, therefore, $A_1\cap A_2$ and $A_1\cup A_2$ belong to $\Sigma^U$. Let $A_n$ be a sequence of
elements of $\Sigma^U$ and $A=\bigcup_{n=1}^\infty A_n$. For all $n=1,2,\ldots$, we set $B_n=\bigcup_{j=1}^n A_j$. Then $B_n\in\Sigma^U$ for all $n$,
and the $\sigma$-additivity of $E$ implies that $E(B_n)$ converge strongly to $E(A)$. Hence, $E(A)$ commutes with $U$, i.e., $A\in\Sigma^U$. We thus
see that $\Sigma^U$ is a $\sigma$-algebra containing all closed sets. This implies that $\Sigma^U$ coincides with the Borel $\sigma$-algebra, and
(\ref{incl1}) is proved.

Inclusions~(\ref{incl}) and~(\ref{incl1}) imply that $(\mathcal X\cup \mathcal X^*)'\subset \mathcal P'_E$ and, hence, $\mathcal A(\mathcal
P_E)\subset \mathcal A(\mathcal X)$. The lemma is proved.
\end{proof}

The next lemma summarizes the facts about Polish and standard Borel spaces that are needed for the proof of Theorem~\ref{t0}.

\begin{lemma}\label{l4a}
\begin{enumerate}
\item[1.] Let $S$ and $S'$ be standard Borel spaces and $f\colon S\to S'$ be a one-to-one measurable mapping. Then $f(S)$ is a measurable subset of
$S'$ and $f$ is a measurable isomorphism from $S$ onto $f(S)$.

\item[2.] Let $S$ be a Polish space and $B$ be its Borel subset. Then there are a Polish space $P$ and a continuous one-to-one map $g\colon P\to S$
such that $B=g(P)$.

\item[3.] If $S$ is a standard Borel space, then there exists a one-to-one function from $S$ to the segment $[0,1]$.
\end{enumerate}
\end{lemma}
\begin{proof}
Statement~1 follows from Theorem~3.2 in~\cite{Mackey}, which, in its turn, is a reformulation of a theorem by Souslin (see~\cite{Kuratowski},
Chapter~III, Sec.~35.IV). For the proof of statement~2, see Lemma~6 of Sec.~IX.6.7 in~\cite{Bourbaki}. To prove statement~3, we recall that every
standard Borel space is either countable or isomorphic to the segment $[0,1]$ (see~\cite{Takesaki}, Appendix, Corollary~A.11). In the latter case,
any isomorphism between $S$ and $[0,1]$ gives us the required function. If $S$ is countable, then we can just choose any one-to-one map from $S$ to
$[0,1]$ because all functions on $S$ are measurable. The lemma is proved.
\end{proof}

\begin{lemma}\label{l4b}
Let $E$ be a spectral measure on a standard Borel space $S$, $I$ be a countable set and $\{f_\iota\}_{\iota\in I}$ be a family of measurable
complex-valued functions on $S$ that separates the points of $S$. Then the von~Neumann algebra generated by all operators $J_{f_\iota}$ with
$\iota\in I$ coincides with $\mathcal A(\mathcal P_E)$.
\end{lemma}
\begin{proof}
Let $f$ denote the map $s\to \{f_\iota(s)\}_{\iota\in I}$ from $S$ to $\C^I$. The space $\C^I$ endowed with its natural product topology is a Polish
space, and the measurability of $f_\iota$ implies that of $f$. Since $f_\iota$ separate the points of $S$, the map $f$ is one-to-one. By statement~1
of Lemma~\ref{l4a}, $f(S)$ is a Borel subset of $\C^I$ and $f$ is a measurable isomorphism of $S$ onto $f(S)$. By statement~2 of Lemma~\ref{l4a},
there are a Polish space $P$ and a continuous one-to-one map $g\colon P\to \C^I$ such that $f(S)=g(P)$. Hence, $h = f^{-1}\circ g$ is a measurable
one-to-one map from $P$ onto $S$. By statement~1 of Lemma~\ref{l4a}, $h$ is a measurable isomorphism from $P$ onto $S$. We now use $h$ to transfer
the topology from $P$ to $S$, i.e., we say that a set $O\subset S$ is open if and only if $h^{-1}(O)$ is open in $P$. Once $S$ is equipped with this
topology, $h$ becomes a homeomorphism between $P$ and $S$ and, hence, $S$ becomes a Polish space. Since $h$ is a measurable isomorphism, the Borel
measurable structure generated by the topology of $S$ coincides with its initial measurable structure. Because $f=g\circ h^{-1}$ is continuous, all
$f_\iota$ are continuous. Hence, the statement follows from Lemma~\ref{l01}. The lemma is proved.
\end{proof}

\begin{proof}[Proof of Theorem~$\ref{t0}$] {$\,$}
\par\noindent

Suppose conditions~(1) and~(2) hold. By Lemma~\ref{l4d}, condition~(1) implies that $\mathcal A(\mathcal X)\subset \mathcal A(\mathcal P_E)$. Let the
family $\{f_\iota\}_{\iota\in I}$ be as in condition~(2) and $\mathcal X_0$ be the set of all $J^E_{f_\iota}$ with $\iota\in I$. Since $E$ is
standard, there is a measurable subset $\tilde S$ of $S$ such that $E(S\setminus\tilde S)=0$ and $\tilde S$, considered as a measurable subspace of
$S$, is a standard Borel space. Let $\tilde E$ denote the restriction of $E$ to $\tilde S$. For each $\iota\in I$, we choose a measurable function
$\tilde f_\iota$ on $\tilde S$ that is equal $E$-a.e. (or, which is the same, $\tilde E$-a.e.) to $f_\iota$. Then we have $J^{\tilde E}_{\tilde
f_\iota}=J^E_{f_\iota}$ for all $\iota\in I$ and it follows from Lemma~\ref{l4b} that $\mathcal A(\mathcal X_0)=\mathcal A(\mathcal P_{\tilde E})$.
As $\mathcal P_E = \mathcal P_{\tilde E}$, this implies that $\mathcal A(\mathcal P_E)\subset \mathcal A(\mathcal X)$ and, hence, $\mathcal
A(\mathcal P_E) = \mathcal A(\mathcal X)$.

Conversely, let $\mathcal A(\mathcal P_E) = \mathcal A(\mathcal X)$. Then condition~(1) is ensured by Lemma~\ref{l4d}. By Lemma~\ref{l41}, there is a
countable set $\mathcal X_0\subset \mathcal X$ such that $\mathcal A(\mathcal X_0)=\mathcal A(\mathcal X)$. Choose a countable family
$\{f_\iota\}_{\iota\in I}$ of measurable complex functions on $S$ such that each $T\in \mathcal X_0$ is equal to $J^E_{f_\iota}$ for some $\iota\in
I$. It suffices to show that $\{f_\iota\}_{\iota\in I}$ is $E$-separating. Let $f$ be the measurable map $s\to \{f_\iota(s)\}_{\iota\in I_0}$ from
$S$ to $\C^I$. For each $\iota\in I$, let $\pi_\iota\colon \C^{I}\to\C$ be the function taking $\{z_\kappa\}_{\kappa\in I}$ to $z_\iota$. For any
$\iota\in I$, we have $\pi_\iota\circ f = f_\iota$, and it follows from~(\ref{pushforward}) that $J^{f_*E}_{\pi_\iota} = J^E_{f_\iota}$ for all
$\iota\in I$, i.e., the set of all $J^{f_*E}_{\pi_\iota}$ coincides with $\mathcal X_0$. Since the family $\{\pi_\iota\}_{\iota\in I}$ obviously
separates the points of $\C^{I}$, Lemma~\ref{l01} implies that $\mathcal A(\mathcal X_0) = \mathcal A(\mathcal P_{f_*E})$ and, hence,
\begin{equation}\label{spec_mes}
\mathcal A(\mathcal P_{E}) = \mathcal A(\mathcal P_{f_*E}).
\end{equation}
By statement~3 of Lemma~\ref{l4a}, there exists a one-to-one measurable function $g$ on $\tilde S$. Clearly, $g$ is $E$-measurable on $S$ and it
follows from Lemma~\ref{l4d} and~(\ref{spec_mes}) that $\mathcal A(J^E_g)\in \mathcal A(\mathcal P_{f_*E})$. Now Lemma~\ref{l4d} implies that there
exists a measurable function $h$ on $\C^{I}$ such that $J^E_g = J^{f_*E}_h$. In view of~(\ref{pushforward}), this means that $J^E_g = J^E_{h\circ f}$
and, hence, $g$ and $h\circ f$ are equal $E$-a.e. Since $g$ is one-to-one on $\tilde S$, it follows that $f$ is one-to-one on $\tilde S\setminus N$,
where $N$ is the set of all $s\in\tilde S$ such that $g(s)\neq h(f(s))$. This means that $\{f_\iota\}_{\iota\in I}$ is $E$-separating and the theorem
is proved.
\end{proof}

\begin{example}
Let $\Lambda\subset \C$ be a set having an accumulation point in $\C$ and let $f_\lambda(z)=e^{\lambda z}$ for $\lambda\in\Lambda$ and $z\in\C$.
Clearly, we can choose a countable set $\Lambda_0\subset \Lambda$ that has an accumulation point in $\C$. If $f_\lambda(z)=f_\lambda(z')$ for some
$z,z'\in \C$ and all $\lambda\in\Lambda_0$, then we have $z=z'$ by the uniqueness theorem for analytic functions and, therefore, the family
$\{f_\lambda\}_{\lambda\in\Lambda_0}$ separates points of $\C$. Theorem~\ref{t0} therefore implies that $\mathcal P_E$ is equivalent to
$\{J^E_{f_\lambda}\}_{\lambda\in\Lambda}$ for any spectral measure $E$ on $\C$. In view of~(\ref{601}), it follows that any normal operator $T$ is
equivalent to the set of all operators $e^{\lambda T}$ with $\lambda\in\Lambda$.
\end{example}

\section{Diagonalizations}
\label{s5}

Let $\nu$ be a measure on a measurable space $S$ and $\mathfrak S$ be a $\nu$-a.e. defined map on $S$ such that $\mathfrak S(s)$ is a Hilbert space
for $\nu$-a.e. $s$ (such a $\mathfrak S$ will be called a $\nu$-a.e. defined family of Hilbert spaces on $S$). A $\nu$-a.e. defined map $\xi$ on $S$
is said to be a $\nu$-a.e. defined section of $\mathfrak G$ if $\xi(s)\in \mathfrak S(s)$ for $\nu$-a.e. $s$. Let $\mathcal F(S,\mathfrak S,\nu)$
denote the set of all equivalence classes whose representatives are $\nu$-a.e. defined sections of $\mathfrak S$. Clearly, $\mathcal F(S,\mathfrak
S,\nu)$ has a natural structure of a vector space. A family $\mathfrak S$ is called $\nu$-measurable on $S$ if a subspace $\mathcal M(\mathfrak S)$
of $\mathcal F(S,\mathfrak S,\nu)$ is chosen such that
\begin{itemize}
\item[(I)] The function $s\to \langle \xi(s),\eta(s)\rangle$ on $S$ is $\nu$-measurable for any sections $\xi,\eta$ of $\mathfrak S$ such that
$[\xi]_\nu, [\eta]_\nu\in \mathcal M(\mathfrak S)$.

\item[(II)] If $\xi$ is a section of $\mathfrak S$ and the function $t\to (\xi(t),\eta(t))$ is $\nu$-measurable for any section $\eta$ of $\mathfrak
S$ such that $[\eta]_\nu\in\mathcal M(\mathfrak S)$, then $[\xi]_\nu\in\mathcal M(\mathfrak S)$.

\item[(III)] There is a sequence $\xi_1,\xi_2,\ldots$ of sections of $\mathfrak S$ such that $[\xi_j]_\nu\in\mathcal M(\mathfrak S)$ for all $j$ and
the linear span of the sequence $\xi_1(s),\xi_2(s),\ldots$ is dense in $\mathfrak S(s)$ for $\nu$-a.e. $s\in S$.
\end{itemize}
Given a $\nu$-measurable family $\mathfrak S$ of Hilbert spaces, a section $\xi$ of $\mathfrak S$ is called $\nu$-measurable if $[\xi]_\nu\in\mathcal
M(\mathfrak S)$. The direct integral $\int^\oplus_S \mathfrak S(s)\,d\nu(s)$ of a $\nu$-measurable family $\mathfrak S$ is, by definition, the vector
subspace of $\mathcal M(\mathfrak S)$ consisting of all $[\xi]_\nu$, where the section $\xi$ is $\nu$-measurable and
\[
\int_S \|\xi(s)\|^2 \,d\nu(s) < \infty.
\]
The scalar product of $[\xi]_\nu, [\eta]_\nu \in \int\nolimits^\oplus_S \mathfrak S(s)\,d\nu(s)$ is defined by the relation
\[
\langle [f]_\nu, [g]_\nu\rangle = \int_S \langle f(s),g(s)\rangle\,d\nu(s).
\]
This scalar product makes $\int^\oplus_S \mathfrak S(s)\,d\nu(s)$ a Hilbert space. A family $\mathfrak S' = \{\mathfrak S'(s)\}_{s\in S}$ of Hilbert
spaces is called a $\nu$-measurable family of subspaces of $\mathfrak S$ if $\mathfrak S'(s)$ is a subspace of $\mathfrak S(s)$ for $\nu$-a.e. $s$
and $\mathcal M(\mathfrak S)\cap \mathcal F(S,\mathfrak S',\nu)$ is a measurable structure for $\mathfrak S'$ (i.e., satisfies conditions (I)-(III)
above). In this case, $\mathfrak S'$ will be always assumed to be endowed with $\mathcal M(\mathfrak S') = \mathcal M(\mathfrak S)\cap \mathcal
F(S,\mathfrak S',\nu)$. A family $\mathfrak S' = \{\mathfrak S'(s)\}_{s\in S}$ of Hilbert spaces is a $\nu$-measurable family of subspaces of
$\mathfrak S$ if and only if $\mathfrak S'(s)$ is a subspace of $\mathfrak S(s)$ for $\nu$-a.e. $s$ and there is a sequence $\xi_1,\xi_2,\ldots$ of
sections of $\mathfrak S'$ such that $[\xi_j]_\nu\in\mathcal M(\mathfrak S)$ for all $j$ and the linear span of the sequence
$\xi_1(s),\xi_2(s),\ldots$ is dense in $\mathfrak S'(s)$ for $\nu$-a.e. $s\in S$.

\begin{example}\label{e_fam}
Let $\mathfrak h$ be a separable Hilbert space and $\nu$ be a measure on a measurable space $S$. Let the family $\mathcal I_{\,\mathfrak h,\nu} =
\{\mathcal I_{\,\mathfrak h,\nu}(s)\}_{s\in S}$ be such that $\mathcal I_{\,\mathfrak h,\nu}(s)=\mathfrak h$ for all $s\in S$ and $\mathcal
M(\mathcal I_{\,\mathfrak h,\nu})$ is the set of all $[\xi]_\nu$, where $\xi$ is a $\nu$-measurable map from $S$ to $\mathfrak h$. It is easy to see
that $\mathcal M(\mathcal I_{\,\mathfrak h,\nu})$ satisfies conditions (I)-(III) and, therefore, $I_{\,\mathfrak h,\nu}$ is a $\nu$-measurable family
of Hilbert spaces.
\end{example}

If $\mathfrak S_1$ and $\mathfrak S_2$ are $\nu$-a.e. defined families of Hilbert spaces on $S$, then the family $\mathfrak S_1\oplus\mathfrak S_2$
is the map defined on $D_{\mathfrak S_1}\cap D_{\mathfrak S_2}$ and taking $s$ to $\mathfrak S_1(s)\oplus\mathfrak S_2(s)$. If both $\mathfrak S_1$
and $\mathfrak S_2$ are $\nu$-measurable, then a $\nu$-a.e. defined section $\xi(s)= (\xi_1(s),\xi_2(s))$ of $\mathfrak S_1\oplus\mathfrak S_2$ is
called $\nu$-measurable if $\xi_1$ and $\xi_2$ are $\nu$-measurable sections of $\mathfrak S_1$ and $\mathfrak S_2$ respectively. Defining $\mathcal
M(\mathfrak S_1\oplus\mathfrak S_2)$ as the set of all $[\xi]_\nu$, where $\xi$ is a $\nu$-measurable section of $\mathfrak S_1\oplus\mathfrak S_2$,
we make $\mathfrak S_1\oplus\mathfrak S_2$ a $\nu$-measurable family of Hilbert spaces on~$S$.

Let $\mathfrak S$ be a $\nu$-measurable family of Hilbert spaces on $S$ and $a$ be a $\nu$-a.e. defined map on $S$ such that $a(s)$ is an operator in
$\mathfrak S(s)$ for $\nu$-a.e. $s$ (such a map will be called a $\nu$-a.e. defined family of operators in $\mathfrak S$). A family $a$ of operators
in $\mathfrak S$ is called $\nu$-measurable if $a(s)$ are closable for $\nu$-a.e. $s$ and the graphs $G_{\overline{a(s)}}$ of $\overline{a(s)}$
constitute a $\nu$-measurable family of subspaces of $\mathfrak S\oplus \mathfrak S$. The family $a(s)$ is measurable if there is a sequence
$\xi_1,\xi_2,\ldots$ of $\nu$-measurable sections of $\mathfrak S$ such that $\xi_n(s)\in D_{a(s)}$ for all $n$ and $\nu$-a.e. $s$ and the linear
span of the vectors $(\xi_n(s),a(s)\xi_n(s))$ is dense in $G_{a(s)}$ for $\nu$-a.e. $s$. As shown in~\cite{Nussbaum}, if $a(s)$ is a $\nu$-measurable
family of operators in $\mathfrak S$, then the map $s\to a(s)\xi(s)$ is a $\nu$-measurable section of $\mathfrak S$ for any $\nu$-measurable section
$\xi$ of $\mathfrak S$ such that $\xi(s)\in D_{a(s)}$ for $\nu$-a.e. $s$.

Let $\mathfrak S$ be a $\nu$-measurable family of Hilbert spaces on $S$, and $g$ be a complex-valued $\nu$-measurable function on $S$. Then $g$
determines a linear operator $\mathcal T_g$ in $\int^\oplus_S \mathfrak S(s)\,d\nu(s)$ as follows. The domain $D_{\mathcal T_g}$ consists of all
$[f]_\nu\in \int^\oplus_S \mathfrak S(s)\,d\nu(s)$, where $f$ is such that the equivalence class of the section $s\to g(s)f(s)$ belongs to
$\int^\oplus_S \mathfrak S(s)\,d\nu(s)$, and the vector $\mathcal T_g [f]_\nu$ is defined as the equivalence class of the section $s\to g(s)f(s)$.

\begin{definition}\label{df-diag}
Let $\mathfrak H$ be a Hilbert space and $\mathcal X$ be a set of closed densely defined operators in $\mathfrak H$. Let $S$ be a measurable space,
$\nu$ be a measure on $S$, $\mathfrak S$ be a $\nu$-measurable family of Hilbert spaces on $S$, and $V\colon \mathfrak H \to \int^\oplus_S \mathfrak
S(s)\,d\nu(s)$ be a unitary operator. We say that the quadruple $(S,\mathfrak S,\nu,V)$ is a diagonalization for $\mathcal X$ if every $T\in\mathcal
X$ is equal to $V^{-1}\mathcal T_g V$ for some $\nu$-measurable complex function $g$ on $S$. A diagonalization $(S,\mathfrak h,\nu,V)$ is called
exact if $\nu$ is standard and $V^{-1}\mathcal T_g V\in \mathcal A(\mathcal X)$ for any $\nu$-measurable $\nu$-essentially bounded function $g$ on
$S$.
\end{definition}

It is easy to see that the above definition of an exact diagonalization is just a reformulation of the condition~(E) given in Introduction in terms
of von~Neumann algebras.

For any $\nu$-measurable family $\mathfrak S$ of Hilbert spaces on $S$, we can define a spectral measure $\Pi_{\mathfrak S}$ on $S$ by setting
\[
\Pi_{\mathfrak S}(A) = \mathcal T_{\chi_A}
\]
for any measurable set $A$, where $\chi_A$ is the characteristic function of $A$. It is easy to see that $J^{\Pi_{\mathfrak S}}_g = \mathcal T_g$ for
any $\nu$-measurable function $g$ on $S$.

\begin{theorem}\label{t0a}
Let $\mathfrak H$ be a separable Hilbert space and $\mathcal X$ be a set of closed densely defined operators in $\mathfrak H$. A quadruple
$(S,\mathfrak S,\nu,V)$, where $S$, $\nu$, $\mathfrak S$, and $V$ are as in Definition~$\ref{df-diag}$, is a diagonalization for $\mathcal X$ if and
only if it is a diagonalization for $\mathcal A(\mathcal X)$. A diagonalization $(S,\mathfrak S,\nu,V)$ for $\mathcal X$ is exact if and only if
$\nu$ is standard and there is a $\nu$-separating family $\{g_\iota\}_{\iota\in I}$ of $\nu$-measurable complex functions on $S$ such that
\begin{equation}\label{diag}
V^{-1}\mathcal T_{g_\iota} V\in \mathcal X,\quad \iota\in I.
\end{equation}
\end{theorem}
\begin{proof}
Let $S$, $\mathfrak S$, $\nu$, and $V$ be as in Definition~\ref{df-diag} and let the spectral measure $E$ for $(\mathfrak H, S)$ be defined by the
relation $E(A)=V^{-1} \Pi_{\mathfrak S}(A) V$, where $A$ is a measurable subset of $S$. Then $E$-measurability coincides with $\nu$-measurability,
and we have
\begin{equation}\label{spec_integral}
J^E_g = V^{-1} J^{\Pi_{\mathfrak S}}_g V= V^{-1}\mathcal T_g V
\end{equation}
for any $\nu$-measurable $g$ on $S$. By Lemmas~\ref{l4d} and~\ref{l4}, every $T\in \mathcal X$ is equal to $J^E_g$ for some $\nu$-measurable $g$ on
$S$ if and only if
\begin{equation}\label{inclu}
\mathcal A(\mathcal X)\subset \mathcal A(\mathcal P_E).
\end{equation}
In view of~(\ref{spec_integral}), this means that the quadruple $(S,\mathfrak S,\nu,V)$ is a diagonalization for $\mathcal X$ if and only if
(\ref{inclu}) holds. Since $\mathcal A(\mathcal A(\mathcal X))=\mathcal A(\mathcal X)$, it follows that the quadruple $(S,\mathfrak S,\nu,V)$ is a
diagonalization for $\mathcal X$ if and only if it is a diagonalization for $\mathcal A(\mathcal X)$.

Let $(S,\mathfrak S,\nu,V)$ be a diagonalization for $\mathcal X$ and $\nu$ be standard. The condition that $V^{-1} \mathcal T_g V\in\mathcal
A(\mathcal X)$ for any $\nu$-essentially bounded $\nu$-measurable function $g$ on $S$ is equivalent to the equality
\begin{equation}\label{equality}
\mathcal A(\mathcal X)= \mathcal A(\mathcal P_E),
\end{equation}
as follows from Lemma~\ref{l4d}, (\ref{spec_integral}), and~(\ref{inclu}). Since $(S,\mathfrak S,\nu,V)$ is a diagonalization for $\mathcal X$, it
follows from~(\ref{spec_integral}) that every $T\in \mathcal X$ is equal to $J^E_g$ for some $\nu$-measurable $g$ on $S$. Now Theorem~\ref{t0}
implies that (\ref{equality}) holds if and only if there is a $\nu$-separating family $\{g_\iota\}_{\iota\in I}$ such that $J^E_{g_\iota}\in \mathcal
X$ for all $\iota\in I$. By~(\ref{spec_integral}), the latter condition is equivalent to~(\ref{diag}). The theorem is proved.
\end{proof}
\begin{corollary}
Let $\mathfrak H$ be a separable Hilbert space and $\mathcal X$ and $\mathcal Y$ be equivalent sets of closed densely defined operators in $\mathfrak
H$. Then every (exact) diagonalization for $\mathcal X$ is also an (exact) diagonalization for $\mathcal Y$.
\end{corollary}

\begin{theorem}\label{t_diag}
Every set of pairwise commuting normal operators in a separable Hilbert space admits an exact diagonalization.
\end{theorem}
\begin{proof}
Let $\mathcal X$ be a set of pairwise commuting normal operators in a separable Hilbert space $\mathfrak H$. By Lemma~\ref{l_com}, the algebra
$\mathcal A(\mathcal X)$ is Abelian. By Th\'eor\`eme~2 of Sec.~II.6.2 in~\cite{Dixmier}, there are a finite measure $\nu$ on a compact metrizable
space $S$, a $\nu$-measurable family $\mathfrak S$ of Hilbert spaces, and a unitary operator $V\colon \mathfrak H \to \int^\oplus_S \mathfrak
S(s)\,d\nu(s)$ such that $\mathcal A(\mathcal X)$ coincides with the set of all operators $V^{-1}\mathcal T_g V$, where $g$ is a $\nu$-measurable
$\nu$-essentially bounded function on $S$. It now follows from Theorem~\ref{t0a} and Definition~\ref{df-diag} that $(S,\mathfrak S,\nu,V)$ is an
exact diagonalization for $\mathcal X$. The theorem is proved.
\end{proof}

\section{Symmetry preserving extensions}
\label{s6}

\begin{definition}\label{d4}
Let $\nu$ be a measure on a measurable space $S$, $\mathfrak S$ be a $\nu$-measurable family of Hilbert spaces on $S$, and $a$ be a $\nu$-a.e.
defined family of operators in $\mathfrak S$. Let $D$ be the subspace of $\int^\oplus \mathfrak S(s)\,d\nu(s)$ consisting of all $[\xi]_\nu$, where
$\xi$ is a square-integrable section of $\mathfrak S$ such that $\xi(s)\in D_{a(s)}$ for $\nu$-a.e. $s$ and the section $s\to a(s)\xi(s)$ of
$\mathfrak S$ is square-integrable. The direct integral $\int^{\oplus} a(s)\,d\nu(s)$ of the family $a(s)$ is defined as the linear operator in
$\int^\oplus \mathfrak S(s)\,d\nu(s)$ with domain $D$ taking $[\xi]_\nu\in D$ to the $\nu$-equivalence class of the map $s\to a(s)\xi(s)$.
\end{definition}

It is easy to see that $\int^{\oplus} a(s)\,d\nu(s)$ is closed if $a(s)$ are closed for almost all $s$. Clearly, the operator $\int^{\oplus}
a(s)\,d\nu(s)$ commutes with $\mathcal T_g$ for any $\nu$-essentially bounded function $g$ on $S$. The next statement was proved in~\cite{Nussbaum}.

\begin{lemma}\label{l5}
Let $\nu$ be a measure a measurable space $S$ and $\mathfrak S$ be a $\nu$-measurable family of Hilbert spaces on $S$. Let $A$ be a closed operator
in $\int^\oplus \mathfrak S(s)\,d\nu(s)$ that commutes with $\mathcal T_g$ for any $\nu$-essentially bounded function $g$ on $S$. Then there is a
unique (up to $\nu$-equivalence) $\nu$-measurable family $a(s)$ of closed operators in $\mathfrak S$ such that
\[
A=\int^{\oplus} a(s)\,d\nu(s).
\]
The operator $A$ is self-adjoint if and only if $a(s)$ is self-adjoint for almost every $s$. If $a(s)$ and $\tilde a(s)$ are $\nu$-measurable
families of closed operators in $\mathfrak S$, then $\int^{\oplus} \tilde a(s)\,d\nu(s)$ is an extension of $\int^{\oplus} a(s)\,d\nu(s)$ if and only
if $\tilde a(s)$ is an extension of $a(s)$ for almost all $s$.
\end{lemma}

\begin{definition}\label{d5}
Let $\nu$ be a measure a measurable space $S$ and $\mathfrak S$ be a $\nu$-measurable family of Hilbert spaces on $S$. A $\nu$-a.e. defined family
$a(s)$ of operators in $\mathfrak S$ is said to be compatible with a subspace $\Delta$ of $\int^\oplus \mathfrak S(s)\,d\nu(s)$ if
\begin{itemize}
\item[(a)] for any $[\xi]_\nu\in \Delta$, the relation $\xi(s)\in D_{a(s)}$ holds for $\nu$-a.e. $s$,

\item[(b)] there is a sequence $\xi_1,\xi_2,\ldots$ of $\nu$-measurable sections of $\mathfrak S$ such that $[\xi_j]_\nu\in\Delta$ for all $n$ and
the linear span of $(\xi_j(s), a(s)\xi_j(s))$ is dense in the graph of $a(s)$ for almost all $s$.
\end{itemize}
\end{definition}

\begin{definition}\label{d_red}
Let $H$ be an operator in a Hilbert space $\mathfrak H$, $\nu$ be a measure a measurable space $S$, and $\mathfrak S$ be a $\nu$-measurable family of
Hilbert spaces on $S$. Let $V$ be a unitary operator from $\mathfrak H$ to $\int^\oplus_S \mathfrak S(s)\,d\nu(s)$. A $\nu$-a.e. defined family
$a(s)$ of operators in $\mathfrak S$ is called a reduction of $H$ with respect to the quadruple $(S,\mathfrak S,\nu,V)$ if the family $a(s)$ is
compatible with $V(D_H)$ and $\int^\oplus_S a(s)\,d\nu(s)$ is an extension of $VHV^{-1}$.
\end{definition}

Clearly, a $\nu$-a.e. defined family $a(s)$ of operators in $\mathfrak S$ is a reduction of $H$ with respect to $(S,\mathfrak S,\nu,V)$ if and only
if $a(s)$ is compatible with $V(D_H)$ and the relation
\begin{equation}\label{sep7}
a(s)\xi(s) = \eta(s)
\end{equation}
holds $\nu$-a.e. for any square-integrable sections $\xi$ and $\eta$ of $\mathfrak S$ such that
\begin{equation}\label{sep7a}
[\xi]_\nu = V\Psi,\quad [\eta]_\nu = VH\Psi
\end{equation}
for some $\Psi\in D_H$.

\begin{theorem} \label{t1}
Let $\mathcal X$ be a set of closed densely defined operators in a Hilbert space $\mathfrak H$ and $(S,\mathfrak S,\nu,V)$ be an exact
diagonalization for $\mathcal X$. Let $H$ be an operator in $\mathfrak H$ with the dense domain $D_H$ and $a(s)$ be a reduction of $H$ with respect
to $(S,\mathfrak S,\nu,V)$. Then the following statements are valid:
\par\bigskip
{\rm 1.} Let $\tilde a(s)$ be a $\nu$-measurable family of closed extensions of $a(s)$. Then
\begin{equation}\label{sep8}
\tilde H = V^{-1}\int^\oplus \tilde a(s) \,d\nu(s)\, V
\end{equation}
is a closed extension of $H$ that commutes with all elements of $\mathcal A(\mathcal X)$. If $\tilde a(s)$ are self-adjoint for $\nu$-a.e. $s$, then
$\tilde H$ is self-adjoint and commutes with all normal operators $T$ such that $\mathcal A(T)\subset \mathcal A(\mathcal X)$, in particular, with
all elements of $\mathcal X$.
\par\bigskip
{\rm 2.} Let $\mathcal Y$ be a set of closed densely defined operators that is equivalent to $\mathcal X$ and $\tilde H$ be a closed extension of $H$
such that both $\tilde H$ and $\tilde H^*$ commute with all elements of $\mathcal Y$. Then there is a unique (up to $\nu$-equivalence)
$\nu$-measurable family $\tilde a(s)$ of closed extensions of $a(s)$ such that formula~$(\ref{sep8})$ holds. If $\tilde H$ is self-adjoint, then
$\tilde a(s)$ are self-adjoint for $\nu$-a.e. $s$.
\par\bigskip
{\rm 3.} Suppose the family $a(s)$ is $\nu$-measurable and there is an involutive set $\mathcal Y\subset L(\mathfrak H)$ that is equivalent to
$\mathcal X$ and leaves $D_H$ invariant (i.e., $T\Psi\in D_H$ for any $\Psi\in D_H$ and $T\in \mathcal Y$). Then $H$ is closable and commutes with
all elements of $\mathcal Y$. Moreover, the operators $a(s)$ are closable for $\nu$-a.e. $s$ and the closure $\bar H$ of $H$ is given by
\begin{equation}\label{sep8a}
\bar H = V^{-1}\int^\oplus \bar a(s)\,d\nu(s)\,V.
\end{equation}
\end{theorem}

\begin{proof} $\,$
\par \noindent 1.
Let $A=\int^{\oplus} a(s)\,d\nu(s)$ and $\tilde A=\int^{\oplus} \tilde a(s)\,d\nu(s)$. By the hypothesis, $A$ is an extension of $VHV^{-1}$. Since
$\tilde A$ is a closed extension of $A$, we conclude that $\tilde H$ is a closed extension of $H$. Theorem~\ref{t0a} implies that $(S,\mathfrak
S,\nu,V)$ is a diagonalization for $\mathcal A(\mathcal X)$. This means that every $T\in \mathcal A(\mathcal X)$ has the form $T= V^{-1}\mathcal T_g
V$ for some $\nu$-essentially bounded $g$ on $S$. Since $\mathcal T_g$ commutes with $\tilde A$, it follows that $T$ commutes with $\tilde H$. If
$\tilde a(s)$ are self-adjoint for $\nu$-a.e. $s$, then $\tilde A$ is self-adjoint by Lemma~\ref{l5} and, therefore, $\tilde H$ is self-adjoint. By
Lemma~\ref{lll}, $\tilde H$ commutes with all normal operators $T$ such that $\mathcal A(T)\subset \mathcal A(\mathcal X)$.

\par\noindent 2. Let $\tilde A=V\tilde H\,V^{-1}$. Since both $\tilde H$ and $\tilde H^*$ commute with all elements of $\mathcal Y$, it
follows from Lemma~\ref{l3a} that $\tilde H$ commutes with all elements of $\mathcal A(\mathcal Y)=\mathcal A(\mathcal X)$. As the diagonalization
$(S,\mathfrak S,\nu,V)$ is exact, we have $V^{-1}\mathcal T_g V\in \mathcal A(\mathcal X)$ for any $\nu$-measurable $\nu$-essentially bounded
function $g$ on $S$. Hence, $\tilde A$ commutes with $\mathcal T_{g}$ for any $\nu$-measurable $\nu$-essentially bounded $g$, and it follows from
Lemma~\ref{l5} that there is a unique (up to $\nu$-equivalence) $\nu$-measurable family $\tilde a$ of closed operators in $\mathfrak S$ such that
\begin{equation}\label{sep10}
\tilde A = \int^{\oplus} \tilde a(s)\,d\nu(s).
\end{equation}
If $\tilde H$ is self-adjoint, then $\tilde A$ is also self-adjoint, and Lemma~\ref{l5} ensures that $\tilde a(s)$ are self-adjoint for $\nu$-a.e.
$s$.

For any $\nu$-measurable section $\xi$ of $\mathfrak S$ such that $[\xi]_\nu\in V(D_H)$, we have $\xi(s)\in D_{a(s)}\cap D_{\tilde a(s)}$ and the
relation
\begin{equation}\label{sep12}
a(s)\xi(s)=\tilde a(s)\xi(s)
\end{equation}
holds for $\nu$-a.e. $s$. Indeed, since the family $a(s)$ is compatible with to $V(D_H)$, we have $\xi(s)\in D_{a(s)}$ for $\nu$-a.e. $s$. Let $\eta$
be a $\nu$-measurable section of $\mathfrak S$ such that $[\eta]_\nu = VHV^{-1}[\xi]_\nu$. Then obviously there is $\Psi\in D_H$ such that
equalities~(\ref{sep7a}) hold and, therefore, relation~(\ref{sep7}) holds for $\nu$-a.e. $s$. Since $\tilde H$ is an extension of $H$, $\tilde A$ is
an extension of $VHV^{-1}$ and, in view of~(\ref{sep10}), we conclude that $\xi(s)\in D_{\tilde a(s)}$ and
\[
\eta(s) = \tilde a(s)\,\xi(s)
\]
for $\nu$-a.e. $s$. Together with~(\ref{sep7}), this equality implies~(\ref{sep12}).

Let $\xi_1,\xi_2,\ldots$ be a sequence of sections of $\mathfrak S$ such that $[\xi_j]_\nu\in V(D_H)$ for all $j$ and the linear span of
$(\xi_j(s),a(s)\xi_j(s))$ is dense in the graph of $a(s)$ for $\nu$-a.e. $s$ (such a sequence exists because the family of operators $a(s)$ is
compatible with to $V(D_H)$). Let $L_s$ denote the linear span of $\xi_j(s)$. By~(\ref{sep12}), $L_s\subset D_{a(s)}\cap D_{\tilde a(s)}$ and $\tilde
a(s)\psi = a(s)\psi$ for any $\psi\in L_s$. For any $\psi\in D_{a(s)}$, there is a sequence $\psi_n$ of elements of $L_s$ such that $\psi_n\to \psi$
and $a(s)\psi_n\to a(s)\psi$ as $n\to\infty$. Since $\tilde a(s)\psi_n = a(s)\psi_n$ for all $n$ and $\tilde a(s)$ is closed, we conclude that
$\psi\in D_{\tilde a(s)}$ and $\tilde a(s)\psi = a(s)\psi$. Hence, $\tilde a(s)$ is an extension of $a(s)$ for $\nu$-a.e. $s$.

\par \noindent 3.
The $\nu$-measurability of the family $a(s)$ implies that $a(s)$ are closable for $\nu$-a.e. $s$ and the family $\bar a(s)$ is $\nu$-measurable. Let
$A=\int^\oplus a(s)\,d\nu(s)$ and $B=\int^\oplus \bar a(s)\,d\nu(s)$. Since $A$ is an extension of $VHV^{-1}$ and $B$ is a closed extension of $A$,
we conclude that $V^{-1}BV$ is a closed extension of $H$. Hence, $H$ is closable. By statement~1, $V^{-1}BV$ commutes with all elements of $\mathcal
A(\mathcal X)$ and, in particular, with all elements of $\mathcal Y$. Since $V^{-1}BV$ is an extension of $H$ and $D_H$ is invariant under $\mathcal
Y$, it follows that $H$ commutes with all elements of $\mathcal Y$, and Lemma~\ref{l1} implies that $\bar H$ commutes with all elements of $\mathcal
Y$. By statement~2, there is a $\nu$-measurable field $\tilde a(s)$ of closed extensions of $a(s)$ such that $\bar H = V^{-1}\int^\oplus \tilde a(s)
\,d\nu(s)\, V$. Since $\tilde a(s)$ are closed for $\nu$-a.e. $s$, it follows that $\tilde a(s)$ are extensions of $\bar a(s)$ for $\nu$-a.e. $s$. On
the other hand, since $V^{-1}BV$ is a closed extension of $H$, $B$ is an extension of $V\bar H V^{-1}= \int^\oplus \tilde a(s) \,d\nu(s)$, and
Lemma~\ref{l5} implies that $\bar a(s)$ is an extension of $\tilde a(s)$ for $\nu$-a.e. $s$. Hence, $\bar a(s)=\tilde a(s)$ for $\nu$-a.e. $s$. The
theorem is proved.
\end{proof}

Note that if the conditions of statement~3 of Theorem~\ref{t1} are satisfied and $a(s)$ is essentially self-adjoint for almost all $s$, then $H$ is
essentially self-adjoint by~(\ref{sep8a}) and Lemma~\ref{l5}.

\begin{lemma}\label{l_red}
Let $H$ be a closable densely defined operator in a Hilbert space $\mathfrak H$ and $\mathcal X\subset L(\mathfrak H)$ be an involutive set of
operators. Let $(S,\mathfrak S,\nu,V)$ be an exact diagonalization for $\mathcal X$. Then $H$ commutes with all elements of $\mathcal X$ if and only
if there exists a $\nu$-measurable reduction of $H$ with respect to $(S,\mathfrak S,\nu,V)$ and $D_H$ is left invariant by all elements of $\mathcal
X$.
\end{lemma}
\begin{proof}
If there is a $\nu$-measurable reduction of $H$ and $D_H$ is left invariant by all elements of $\mathcal X$, then $H$ commutes with all elements of
$\mathcal X$ by statement~3 of Theorem~\ref{t1}. If $H$ commutes with all elements of $\mathcal X$, then $\bar H$ also commutes with them by
Lemma~\ref{l1}. In view of the involutivity of $\mathcal X$, Lemma~\ref{l3a} implies that $\bar H$ commutes with all elements of $\mathcal A(\mathcal
X)$. As the diagonalization $(S,\mathfrak S,\nu,V)$ is exact, we have $V^{-1}\mathcal T_g V\in \mathcal A(\mathcal X)$ for any $\nu$-measurable
$\nu$-essentially bounded function $g$ on $S$. Hence, $V\bar H V^{-1}$ commutes with $\mathcal T_{g}$ for any $\nu$-measurable $\nu$-essentially
bounded $g$, and it follows from Lemma~\ref{l5} that there is a unique (up to $\nu$-equivalence) $\nu$-measurable family $a(s)$ of closed operators
in $\mathfrak S$ such that
\[
V\bar H V^{-1} = \int^{\oplus} a(s)\,d\nu(s).
\]
It suffices to show that $a(s)$ is compatible with $V(D_H)$. Let $\xi_1,\xi_2,\ldots$ be square-integrable sections of $\mathfrak S$ such that
$[\xi_j]_\nu\in V(D_H)$ for all $j$ and the linear span of the sequence $([\xi_j]_\nu, VHV^{-1}[\xi_j]_\nu)$ is dense in the graph of $VHV^{-1}$.
Then this sequence is also dense in the graph of $V\bar H V^{-1}$, and it follows from Proposition~8 of Sec.~II.1.6 in~\cite{Dixmier} that $\xi_j$
satisfy the conditions of Definition~\ref{d5}. The lemma is proved.
\end{proof}

Let $\mathfrak h$ be a separable Hilbert space, $\nu$ be a measure on a measurable space $S$, and the $\nu$-measurable family $I_{\,\mathfrak h,\nu}$
of Hilbert spaces on $S$ be as in Example~\ref{e_fam}. In this case, the direct integral $\int^\oplus I_{\,\mathfrak h,\nu}(s)\,d\nu(s)$ obviously
coincides with $L^2(S,\mathfrak h,\nu)$. For brevity, we shall speak of families of subspaces of $\mathfrak h$ and families of operators in
$\mathfrak h$ when referring to families of subspaces of $I_{\,\mathfrak h,\nu}$ and families of operators in $I_{\,\mathfrak h,\nu}$ respectively.
Similarly, we shall say that $(S,\mathfrak h,\nu,V)$ is a diagonalization for a set of operators $\mathcal X$ if $(S,I_{\,\mathfrak h,\nu},\nu,V)$ is
a diagonalization for $\mathcal X$.

\begin{definition} \label{d2}
Let $\mathfrak h$ be a Hilbert space, $\nu$ be a measure on a set $S$, and $D$ be a linear subspace of $\mathfrak h$. We say that a $\nu$-a.e.
defined family $a(s)$ of operators in $\mathfrak h$ is $\nu$-regular with respect to $D$ if $D_{a(s)}= D$ for $\nu$-a.e. $s$ and there is a countable
subset $Y$ of $D$ such that the linear span of the elements $(\psi,a(s)\psi)$ with $\psi\in Y$ is dense in the graph of $a(s)$ for $\nu$-a.e. $s$.
\end{definition}

For any $\psi\in \mathfrak h$ and $f\in L^2(S,d\nu)$, we define $\Phi_{\psi,f}\in L^2(S,\mathfrak h,d\nu)$ by the relation
\[
\Phi_{\psi,f}(s) = f(s)\psi
\]
for almost all $s\in S$.

We say that a sequence $g_1,g_2,\ldots$ of $\nu$-a.e. defined complex-valued functions on $S$ is $\nu$-nonvanishing if there are a $\nu$-null set $N$
such that $S\setminus N$ is contained in the domains of definition of all $g_j$ and for any $s\in S\setminus N$, the condition $\mathbf g_j(s)\neq 0$
is satisfied for some $j$.

\begin{lemma}\label{l6}
Let $\mathfrak h$ be a Hilbert space, $D$ be a linear subspace of $\mathfrak h$, $\nu$ be a measure on a set $S$, and $a(s)$ be a $\nu$-a.e. defined
family of operators in $\mathfrak h$ which is $\nu$-regular with respect to $D$. Then the following statements hold:
\par\medskip\noindent $1.$ If the map $s\to a(s)\psi$ is $\nu$-measurable for every $\psi\in D$, then the family $a(s)$ is $\nu$-measurable.

\par\medskip\noindent $2.$ Let $\Delta$ be a subspace of $L^2(S,\mathfrak h,d\nu)$. Suppose for any $\psi\in D$, there is a $\nu$-nonvanishing sequence
$g_1,g_2,\ldots$ of square-integrable functions such that the $\nu$-equivalence classes of maps $s\to g_j(s)\psi$ belong to $\Delta$ for all $j$.
Then $a(s)$ is compatible with $\Delta$.
\end{lemma}
\begin{proof}
$\,$
\par \noindent 1.
Let $Y\subset D$ satisfy the conditions of Definition~\ref{d2}. We enumerate the elements of $Y$ as a sequence $\psi_1,\psi_2,\ldots$. For each
$n=1,2,\ldots$, let the map $\xi_n$ on $S$ be defined by the relation $\xi_n(s)=(\psi_n,a(s)\psi_n)$. Clearly, $\xi_n$ are $\nu$-measurable maps from
$S$ to $\mathfrak h\oplus \mathfrak h$ for all $n$, and Definition~\ref{d2} implies that the linear span of $\xi_n(s)$ is dense in $\mathfrak h$ for
$\nu$-a.e. $s$. Hence, $a(s)$ is $\nu$-measurable.

\par \noindent 2. For each $n=1,2,\ldots$, we choose a $\nu$-nonvanishing sequence $g^{(n)}_j$ of square-integrable functions such that the
$\nu$-equivalence classes of all maps $\eta^{(n)}_j(s)=g^{(n)}_j(s)\psi_n$ on $S$ belong to $\Delta$. Then for $\nu$-a.e. $s$, the elements
$(\eta^{(n)}_j(s),a(s)\eta^{(n)}_j(s))$ have the same linear span as $(\psi_n,a(s)\psi_n)$, which is dense in the graph of $a(s)$ by the
$\nu$-regularity of the family $a$. The lemma is proved.
\end{proof}

\section{Measurable families of one-dimensional Schr\"odinger operators}
\label{s_meas}

Let $-\infty \leq a < b \leq \infty$ and $\lambda$ be the Lebesgue measure on $(a,b)$. Let $q$ be a locally $\lambda$-square-integrable real function
on $(a,b)$. Let $\mathcal D$ denote the space of all absolutely continuous functions on $(a,b)$ whose derivative is also absolutely continuous. For
$f\in \mathcal D$, we denote by $l_q f$ the $\lambda$-equivalence class of the function
\[
x\to -f''(x)+ q(x) f(x).
\]
Clearly, $l_q$ is a linear operator from $\mathcal D$ to the space of complex $\lambda$-equivalence classes on $(a,b)$. Let\footnote{Throughout this
section, all equivalence classes will be taken with respect to the restriction $\lambda_{(a,b)}$ of $\lambda$ to $(a,b)$. We shall drop the subscript
and write $[f]$ instead of $[f]_{\lambda_{(a,b)}}$.} $\mathcal D_q = \{f\in \mathcal D : [f] \mbox{ and } l_q f \mbox{ are both in } L^2(a,b)\}$ and
$D_q$ be the space of all equivalence classes $[f]$ with $f\in \mathcal D_q$. Let $D_0$ be the space of all $[f]$, where $f$ belongs to the space
$C_0^\infty(a,b)$ of smooth functions whose support is compact and contained in $(a,b)$. We obviously have $D_0\subset D_q\subset L^2(a,b)$. We
define the operator $L_q^*$ in $L^2(a,b)$ by the relations
\begin{align}
& D_{L_q^*} = D_q,\nonumber\\
& L_q^* [f] = l_q f, \quad f\in \mathcal D_q.\nonumber
\end{align}
The operator $L_q$ in $L^2(a,b)$ is defined as the restriction of $L^*_q$ to $D_0$. Then $L_q$ is a symmetric operator and its adjoint is $L_q^*$
(this justifies our notation). For any $f,g\in \mathcal D$, their Wronskian $W(f,g)$ is an absolutely continuous function on $(a,b)$ defined by the
relation
\begin{equation}\label{wronskian}
W(f,g)(x) = f(x)g'(x) - f'(x)g(x).
\end{equation}
A $\lambda$-measurable function $f$ on $(a,b)$ is said to be left (right) square-integrable if $\int_a^c |f(x)|^2\,dx <\infty$ (resp., $\int_c^b
|f(x)|^2\,dx <\infty$) for any $c\in (a,b)$. If $f, g\in \mathcal D$ are left square-integrable functions such that $l_q f$ and $l_q g$ are also left
square-integrable, then the following limit exist:
\begin{equation}\label{limits}
W(f,g)(a) = \lim_{x\downarrow a} W(f,g)(x).
\end{equation}
Similarly, the limit
\[
W(f,g)(b) = \lim_{x\uparrow b} W(f,g)(x)
\]
exists for any right square-integrable $f,g\in \mathcal D$ such that $l_q f$ and $l_q g$ are also right square-integrable. The closure $\bar L_q$ of
$L_q$ is the restriction of $L^*_q$ to the subspace
\begin{equation}\label{clos}
D_{\bar L_q} = \{ [f] : f\in \mathcal D_q \mbox{ and } W(f,g)(a)=W(f,g)(b)=0 \mbox{ for any } g\in \mathcal D_q\}.
\end{equation}

We now consider the homogeneous equation
\begin{equation}\label{hom_eq}
l_q f =0, \quad f\in \mathcal D.
\end{equation}
There are two possibilities
\begin{enumerate}
  \item All solutions of~(\ref{hom_eq}) are left square-integrable (the limit circle case (lcc) at $a$).
  \item There is a solution of~(\ref{hom_eq}) that is not left square-integrable (the limit point case (lpc) at $a$).
\end{enumerate}
The analogous alternative holds for the right end $b$ of the interval. If $f$ and $g$ are solutions of~(\ref{hom_eq}), then the function $W(f,g)(x)$
does not depend on $x$. It is nonzero if and only if $f$ and $g$ are linearly independent. The lpc holds at $b$ (at $a$) if and only if the condition
\begin{equation}\label{lpc}
W(f,g)(b) =0\quad (\mbox{resp., } W(f,g)(a)=0)
\end{equation}
is satisfied for any $f,g\in \mathcal D_q$.

The description of the self-adjoint extensions of $L_q$ depends on whether we have the limit point or limit circle case at the ends of the interval.
In what follows, we assume that lpc holds at $b$. If lpc holds at $a$, then $L_q$ is essentially self-adjoint. If lcc holds at $a$, then the
self-adjoint extensions of $L_q$ are parametrized by the real nontrivial solutions of~(\ref{hom_eq}) and can be described as follows. Given a real
nontrivial solution $f$ of~(\ref{hom_eq}), let
\[
\mathcal D^f_q = \{ g\in \mathcal D_q : W(f,g)(a)=0 \}
\]
and let $D^f_q$ denote the space of all equivalence classes $[f]$ with $f\in \mathcal D^f_q$. Then the restriction $L^f_q$ of $L^*_q$ to $D^f_q$ is a
self-adjoint extension of $L_q$. Moreover, all self-adjoint extensions of $L_q$ can be obtained in this way. Given two real nontrivial solutions $f$
and $\tilde f$ of~(\ref{hom_eq}), we have $L^f_q=L^{\tilde f}_q$ if and only if $f=C\tilde f$, where $C$ is a real number. In the lcc at $a$, the
deficiency indices are $(1,1)$. This implies, in particular, that the orthogonal complement $G_{T}\ominus G_{\bar L_q}$ of the graph $G_{\bar L_q}$
of $L_q$ in the graph $G_{T}$ of $T$ is one-dimensional for any self-adjoint extension $T$ of $L_q$. Let $f_1$ and $f_2$ be linearly independent
solutions of~(\ref{hom_eq}) and let $g\in \mathcal D_q$. Let $\varphi$ be a $\lambda$-measurable function such that $[\varphi]=l_q g$. Then the
function
\begin{equation}\label{zeta}
\rho_g(x) = \frac{1}{W(f_1,f_2)}\left[ f_1(x)\int_a^x \varphi(x')f_2(x')\,dx' - f_2(x)\int_a^x \varphi(x')f_1(x')\,dx'\right]
\end{equation}
belongs to $\mathcal D$ and satisfies the equation
\[
l_q \rho_g = l_q g.
\]
Hence, the function
\[
\sigma_g = g -\rho_g
\]
is a solution of~(\ref{hom_eq}). It is straightforward to check that $\rho_g$ and $\sigma_g$ do not depend on the choice of the solutions $f_1$ and
$f_2$. In particular, we can choose $f_1$ and $f_2$ to be real. Hence, if $g$ is real, then $\rho_g$ and $\sigma_g$ are real.

\begin{lemma}\label{l_ext}
Suppose lcc holds at $a$. Let $g\in \mathcal D_q$ be a real function and $T$ be a self-adjoint extension of $L_q$. Then $[g]\in D_{T}\setminus
D_{\bar L_q}$ if and only if $\sigma_g$ is nontrivial and $T = L_q^{\sigma_g}$.
\end{lemma}
\begin{proof}
It follows easily from~(\ref{zeta}) that
\begin{multline}
W(\rho_g,h)(x) = \frac{1}{W(f_1,f_2)}\left[ W(f_1,h)(x)\int_a^x \varphi(x')f_2(x')\,dx' -\right.\\\left.- W(f_2,h)(x)\int_a^x
\varphi(x')f_1(x')\,dx'\right]\nonumber
\end{multline}
for any $h\in \mathcal D$, where $[\varphi] = l_q g$ and $f_1,f_2$ are linearly independent solutions of~(\ref{hom_eq}). This implies that
\begin{equation}\label{wr}
W(\rho_g,h)(a) =0
\end{equation}
for any $h\in \mathcal D$ such that $h$ and $l_q h$ are left square-integrable and, therefore,
\begin{equation}\label{wr1}
W(g,h)(a) = W(\sigma_g,h)(a).
\end{equation}
Hence, $\sigma_g$ is trivial if and only if $W(g,h)(a)=0$ for any $h\in\mathcal D_q$. In view of~(\ref{lpc}) (recall that lpc is assumed to hold at
$b$), the latter condition is satisfied if and only if $[g]\in D_{\bar L_q}$.

Suppose now that $[g]\in D_{T}\setminus D_{\bar L_q}$. By the above, $\sigma_g$ is nontrivial. Let $f$ be a real solution of~(\ref{hom_eq}) such that
$T=L^f_q$. Then we have $W(f,g)(a)=0$, and it follows from~(\ref{wr}) that $W(f,\sigma_g)=0$. This means that $\sigma_g = Cf$ for some real $C\neq 0$
and, therefore, $T=L^{\sigma_g}_q$. Conversely, suppose $\sigma_g$ is nontrivial and $T=L^{\sigma_g}_q$. Since $\sigma_g$ is nontrivial, we have
$[g]\notin D_{\bar L_q}$. Setting $h=\sigma_g$ in~(\ref{wr1}), we obtain $W(\sigma_g,g)(a)=0$ and, hence, $[g]\in D_T$. The lemma is proved.
\end{proof}

Given a function $f(s,x)$ of two variables, we denote by $f_{[s]}$ the partial function determined by $f$ for a fixed first argument, i.e., the
domain $D_{f_{[s]}}$ of $f_{[s]}$ consists of all $x$ such that $(s,x)\in D_f$ and
\[
f_{[s]}(x) = f(s,x),\quad x\in D_{f_{[s]}}.
\]

\begin{lemma}\label{l_v}
Let $\nu$ be a measure on a measurable space $S$ and $v$ be a $(\nu\times\lambda)$-measurable real function on $S\times (a,b)$ such that $v_{[s]}$ is
locally square-integrable for $\nu$-a.e. $s$. Then the family $s\to L_{v_{[s]}}$ on $S$ of operators in $L^2(a,b)$ is $\nu$-measurable and
$\nu$-regular with respect to $D_0$.
\end{lemma}
\begin{proof}
Let $C_0^\infty(a,b)$ be endowed with the topology defined by the norms
\[
\|f\|_{K,n} = \sup_{x\in K,\,0\leq j\leq n} |f^{(j)}(x)|,
\]
where $n=0,1,\ldots$, $K$ is a compact subset of $(a,b)$ and $f^{(j)}$ is the $j$-th derivative of $f$. Then $C_0^\infty(a,b)$ becomes a separable
metrizable space such that $l_{v_{[s]}}$ induce continuous linear maps from $C_0^\infty(a,b)$ to $L^2(a,b)$ for $\nu$-a.e. $s$. Since the map $f\to
[f]$ puts $C_0^\infty(a,b)$ and $D_0$ in a one-to-one correspondence, we can transfer the topology from $C_0^\infty(a,b)$ to $D_0$. This makes $D_0$
a separable metrizable space such that $L_{v_{[s]}}$ are continuous maps from $D_0$ to $L^2(a,b)$ for $\nu$-a.e. $s$. It follows that $D=D_0$,
$a(s)=L_{v_{[s]}}$, and an arbitrary countable dense subset $Y$ of $D_0$ satisfy the conditions of Definition~\ref{d2} and, therefore, the family
$L_{v_{[s]}}$ is $\nu$-regular with respect to $D_0$. If $f\in C_0^\infty(a,b)$, then the function $(s,x)\to -f''(x) + v(s,x)f(x)$ on $S\times (a,b)$
is $(\nu\times\lambda)$-measurable. Lemma~\ref{l_meas} hence implies that the map $s\to L_{v_{[s]}}[f]$ is $\nu$-measurable. The $\nu$-measurability
of $L_{v_{[s]}}$ now follows from statement~1 of Lemma~\ref{l6}. The lemma is proved.
\end{proof}

Let $\nu$ be a measure on a measurable space $S$. Given a $(\nu\times\lambda)$-measurable real function $v$ on $S\times (a,b)$ such that $v_{[s]}$ is
locally square-integrable for $\nu$-a.e. $s$, we can consider the homogeneous equation
\begin{equation}\label{main_eq}
l_{v_{[s]}}f_{[s]} =0.
\end{equation}
Let $A$ be a $\nu$-measurable subset of $S$. A $(\nu\times\lambda)$-a.e. defined function $f$ on $A\times (a,b)$ will be called a solution
of~(\ref{main_eq}) on $A$ if $f_{[s]}$ belongs to $\mathcal D$ and satisfies~(\ref{main_eq}) for $\nu$-a.e. $s\in A$. A solution $f$ on $A$ is called
nontrivial if $f_{[s]}\neq 0$ for $\nu$-a.e. $s\in A$. Two solutions $f_1$ and $f_2$ on $A$ are called linearly independent if $(f_1)_{[s]}$ and
$(f_2)_{[s]}$ are linearly independent for $\nu$-a.e. $s\in A$.

\begin{lemma}\label{l_sol}
Let $\nu$ be a measure on a measurable space $S$ and $v$ be as in Lemma~$\ref{l_v}$. Let $f$ be a solution of~$(\ref{main_eq})$ on $S$. Suppose there
is $x_0\in (a,b)$ such that the functions $s\to f_{[s]}(x_0)$ and $s\to f'_{[s]}(x_0)$ are $\nu$-measurable. Then $f$ is
$(\nu\times\lambda)$-measurable.
\end{lemma}
\begin{proof}
We shall show that $f$ is $(\nu\times\lambda)$-measurable on $S\times [x_0,b)$ by proving that it is $(\nu\times\lambda)$-measurable on $S\times
[x_0,c]$ for any $c\in (x_0,b)$. We first assume that there is $0<C<\infty$ such that
\begin{equation}\label{bound}
\int_{x_0}^c |v(s,x)|\,dx < C
\end{equation}
for $\nu$-a.e. $s$. We define the functions $f_0,f_1,\ldots$ on $S\times (a,b)$ by the relations
\begin{align}
& f_0(s,x) = f_{[s]}(x_0) + f'_{[s]}(x_0)(x-y),\nonumber \\
& f_n(s,x) = f_0(s,x) + \int_{x_0}^x dx'\int_{x_0}^{x'} v(s,\xi) f_{n-1}(s,\xi)\,d\xi,\quad n=1,2,\ldots. \label{eqn}
\end{align}
By the hypothesis, the function $f_0$ is $(\nu\times\lambda)$-measurable, and it follows from Lemma~\ref{l_meas1} that $f_n$ are
$(\nu\times\lambda)$-measurable for all $n$. Let $S'\subset S$ be a measurable set with a $\nu$-null complement in $S$ such that $v_{[s]}$ is locally
square-integrable, (\ref{bound}) holds, and $f_{[s]}$ is a solution of~(\ref{main_eq}) for all $s\in S'$. Since $f_{[s]}$ satisfies~(\ref{main_eq}),
we have
\begin{equation}\label{int_eq}
f(s,x) = f_0(s,x) + \int_{x_0}^x dx'\int_{x_0}^{x'} v(s,\xi) f(s,\xi)\,d\xi
\end{equation}
for any $s\in S'$ and any $x\in (a,b)$. Let $x_1\in (x_0,c]$ be such that $x_1-x_0 < 1/C$ and let
\[
M_n(s) = \sup_{x_0\leq x\leq x_1} |f(s,x)-f_n(s,x)|,\quad n=0,1,\ldots.
\]
It follows from~(\ref{eqn}) and~(\ref{int_eq}) that
\[
M_n(s) \leq C(x_1-x_0)M_{n-1}(s),\quad n=1,2,\ldots,
\]
for any $s\in S'$. We hence have
\[
M_n(s) \leq [C(x_1-x_0)]^n M_0(s),\quad s\in S'.
\]
Since $C(x_1-x_0)<1$, this means that $f_n$ converge to $f$ pointwise on $S'\times [x_0,x_1]$ and, therefore, $f$ is $(\nu\times\lambda)$-measurable
on $S\times [x_0,x_1]$. Moreover, it follows from~(\ref{eqn}) and~(\ref{int_eq}) that
\[
\sup_{x_0\leq x\leq x_1} |f'_{[s]}(x)-(f_n)'_{[s]}(x)|\leq C M_n(s)
\]
for any $s\in S'$. Hence, $(f_n)'_{[s]}(x)$ converge to $f'_{[s]}(x)$ for any $s\in S'$ and $x\in [x_0,x_1]$. In particular, the functions $s\to
(f_n)_{[s]}(x_1)$ and $s\to (f_n)'_{[s]}(x_1)$ converge $\nu$-a.e. to the functions $s\to f_{[s]}(x_1)$ and $s\to f'_{[s]}(x_1)$ respectively. This
implies that the latter two functions are $\nu$ measurable because the functions $s\to (f_n)_{[s]}(x_1)$ and $s\to (f_n)'_{[s]}(x_1)$ are
$\nu$-measurable by the Fubini theorem. We therefore can repeat the above arguments replacing $x_0$ with $x_1$ and choosing some $x_2\in (x_1,c]$
such that $C(x_2-x_1)<1$. As a result, we shall prove that $f$ is $(\nu\times\lambda)$-measurable on $S\times [x_0,x_2]$. Obviously, after a finite
number of such steps we shall establish the $(\nu\times\lambda)$-measurability of $f$ on $S\times [x_0,c]$.

In the general case (when (\ref{bound}) does not necessarily hold), we consider, for any $N>0$, the set $A_N$ of all $s\in S$ such that $\int_{x_0}^c
|v(s,x)|\,dx < N$. By the Fubini theorem $A_N$ is $\nu$-measurable. The above arguments show that $f$ is $(\nu\times\lambda)$-measurable on
$A_N\times [x_0,c]$. Since $S\setminus\bigcup_{N=1}^\infty A_N$ is a $\nu$-null set, we conclude that $f$ is $(\nu\times\lambda)$-measurable on
$S\times [x_0,c]$. Hence $f$ is $(\nu\times\lambda)$-measurable on $S\times [x_0,b)$.

Repeating the same proof with obvious changes, we make sure that $f$ is $(\nu\times\lambda)$-measurable on $S\times (a,x_0]$ and, hence on $S\times
(a,b)$. The lemma is proved.
\end{proof}

\begin{corollary}\label{cor_sol}
Let $\nu$, $S$, and $v$ be as in Lemma~$\ref{l_sol}$. Then there are $(\nu\times\lambda)$-measurable real solutions $f_1$ and $f_2$
of~$(\ref{main_eq})$ on $S$ such that $(f_1)_{[s]}$ and $(f_2)_{[s]}$ are linearly independent elements of $\mathcal D$ for $\nu$-a.e. $s$.
\end{corollary}
\begin{proof}
Let $S'\subset S$ be a set with a $\nu$-null complement in $S$ such that $v_{[s]}$ is locally square integrable for all $s\in S'$. Choose $x_0\in
(a,b)$. By Theorem~2 of Chapter~V, Sec.~16 in~\cite{Naimark}, there are real functions $f_1$ and $f_2$ on $S'\times (a,b)$ such that (\ref{main_eq})
holds and the conditions
\[
(f_1)_{[s]}(x_0) = (f_2)'_{[s]}(x_0)=1,\quad (f_1)'_{[s]}(x_0) = (f_2)_{[s]}(x_0)=0
\]
are satisfied for all $s\in S'$. Obviously, $(f_1)_{[s]}$ and $(f_2)_{[s]}$ are linearly independent, and Lemma~\ref{l_sol} implies that $f_1$ and
$f_2$ are $(\nu\times\lambda)$-measurable. The corollary is proved.
\end{proof}

\begin{lemma}\label{l_ext1}
Let $\nu$ be a measure on a measurable space $S$ and let $v$ be a $(\nu\times\lambda)$-measurable real function on $S\times (a,b)$ such that
$v_{[s]}$ is locally square-integrable, lcc holds for $l_{v_{[s]}}$ at $a$, and lpc holds for $l_{v_{[s]}}$ at $b$ for $\nu$-a.e. $s$. If $f$ is a
nontrivial real $(\nu\times\lambda)$-measurable solution of~$(\ref{main_eq})$ on $S$, then $L^{f_{[s]}}_{v_{[s]}}$ is a $\nu$-measurable family of
self-adjoint extensions of $L_{v_{[s]}}$. If $H(s)$ is a $\nu$-measurable family of self-adjoint extensions of $L_{v_{[s]}}$, then there is a
nontrivial real $(\nu\times\lambda)$-measurable solution $f$ of~(\ref{main_eq}) on $S$ such that $H(s) = L^{f_{[s]}}_{v_{[s]}}$ for $\nu$-a.e. $s$.
\end{lemma}
\begin{proof}
By Lemma~\ref{l_v}, the the operators $\bar L_{v_{[s]}}$ constitute a $\nu$-measurable family on $S$. This means that there is a sequence
$\zeta_1,\zeta_2,\ldots$ of $\nu$-measurable maps from $S$ to $L^2(a,b)\oplus L^2(a,b)$ such that the linear span of $\zeta_1(s),\zeta_2(s),\ldots$
is dense in the graph $G_{\bar L_{v_{[s]}}}$ for $\nu$-a.e. $s$.

Let $f$ be a real nontrivial $(\nu\times\lambda)$-measurable solution of~(\ref{main_eq}) on $S$ and $\mathcal L_f(s) = L^{f_{[s]}}_{v_{[s]}}$ for all
$s\in S$. Let $\tau$ be a smooth function on $(a,b)$ that is equal to unity in a neighborhood of $a$ and vanishes in a neighborhood of $b$. Let the
functions $g$ and $h$ on $S\times (a,b)$ be defined by the relations
\begin{align}
& g(s,x) = \tau(x) f(s,x) \nonumber\\
& h(s,x) = -\tau''(x) f(s,x) - 2\tau'(x) f'_{[s]}(x).\nonumber
\end{align}
Clearly, both $g$ and $h$ are $(\nu\times\lambda)$-measurable, $g_{[s]}\in \mathcal D_{v_{[s]}}$ for $\nu$-a.e. $s$, and
\begin{equation}\label{g_s}
l_{v_{[s]}} g_{[s]} = [h_{[s]}]
\end{equation}
for $\nu$-a.e. $s$. Let the maps $\xi$ and $\eta$ from $S$ to $L^2(a,b)$ be defined by the relations
\[
\xi(s) = [g_{[s]}],\quad \eta(s) = [h_{[s]}].
\]
By Lemma~\ref{l_meas}, $\xi$ and $\eta$ are $\nu$-measurable. Since $\sigma_{g_{[s]}} = f_{[s]}$ for $\nu$-a.e. $s$, it follows from
Lemma~\ref{l_ext} that $\xi(s)\in D_{\mathcal L_f(s)}\setminus D_{\bar L_{v_{[s]}}}$. In view of~(\ref{g_s}), this implies that $\mathcal
L_f(s)\xi(s) = \eta(s)$ for $\nu$-a.e. $s$. Hence, the map $\zeta\colon s\to (\xi(s),\eta(s))$ from $S$ to $L^2(a,b)\oplus L^2(a,b)$ is
$\nu$-measurable and $\zeta(s)\in G_{\mathcal L_f(s)}\setminus G_{\bar L_{v_{[s]}}}$ for $\nu$-a.e. $s$. Since $G_{\mathcal L_f(s)}\ominus G_{\bar
L_{v_{[s]}}}$ is one-dimensional for $\nu$-a.e. $s$, this implies that the linear span of the sequence $\zeta(s),\zeta_1(s),\zeta_2(s),\ldots$ is
dense in $G_{\mathcal L_f(s)}$ for $\nu$-a.e. $s$. This means that $\mathcal L_f(s)$ constitute a measurable family of operators on $S$.

Conversely, let $H(s)$ be a $\nu$-measurable family of self-adjoint extensions of $L_{v_{[s]}}$. Then both $G_{H(s)}$ and $G_{\bar L_{v_{[s]}}}$ form
$\nu$-measurable families of subspaces of $L^2(a,b)\oplus L^2(a,b)$ and, therefore, $G_{H(s)}\ominus G_{\bar L_{v_{[s]}}}$ is also a $\nu$-measurable
family of subspaces of $L^2(a,b)\oplus L^2(a,b)$. Since $G_{H(s)}\ominus G_{\bar L_{v_{[s]}}}$ is one-dimensional for $\nu$-a.e. $s$, there is a
$\nu$-measurable map $\zeta(s)=(\xi(s),\eta(s))$ from $S$ to $G_{H(s)}\ominus G_{\bar L_{v_{[s]}}}$ such that $\zeta(s)\neq 0$ for $\nu$-a.e. $s$. We
obviously have
\begin{equation}\label{member}
\xi(s)\in D_{H(s)}\setminus D_{\bar L_{v_{[s]}}}
\end{equation}
for $\nu$-a.e. $s$. Let $g$ be a function on $S\times (a,b)$ such that $g_{[s]}\in \mathcal D_{v_{[s]}}$ and $[g_{[s]}] = \xi(s)$ for $\nu$-a.e. $s$.
Since $\xi$ is $\nu$-measurable, Lemma~\ref{l_meas} implies that $g$ is $(\nu\times\lambda)$-measurable. Let $Q$ be the set of all $s\in S$ such that
$g_{[s]}$ has a nonzero real part. We define the function $\tilde g$ on $S\times (a,b)$ by the relation
\[
\tilde g(s,x) = \left\{
\begin{matrix}
\frac{g(s,x)+\overline{g(s,x)}}{2},& s\in Q,\\
\frac{g(s,x)-\overline{g(s,x)}}{2i},& s\in S\setminus Q.
\end{matrix}
\right.
\]
Then $\tilde g$ is $(\nu\times\lambda)$-measurable and $\tilde g_{[s]}$ is a real element of $\mathcal D_{v_{[s]}}$ for $\nu$-a.e. $s$. In view
of~(\ref{member}), we have $[\tilde g_{[s]}]\in D_{H(s)}\setminus D_{\bar L_{v_{[s]}}}$ for $\nu$-a.e. $s$, and it follows from Lemma~\ref{l_ext}
that $\sigma_{\tilde g_{[s]}}$ is nontrivial and
\begin{equation}\label{hsigma}
H(s) = L^{\sigma_{\tilde g_{[s]}}}_{v_{[s]}}
\end{equation}
for $\nu$-a.e. $s$. By Corollary~\ref{cor_sol}, there are $(\nu\times\lambda)$-measurable solutions $f_1$ and $f_2$ of~(\ref{main_eq}) on $S$ such
that $(f_1)_{[s]}$ and $(f_2)_{[s]}$ are linearly independent elements of $\mathcal D$ for $\nu$-a.e. $s$. Let the function $f$ on $S\times (a,b)$ be
given by
\begin{multline}
f(s,x) = \tilde g(s,x) -\\- \frac{1}{W(s)}\left[ f_1(s,x)\int_a^x \varphi(s,x')f_2(s,x')\,dx' - f_2(s,x)\int_a^x \varphi(s,x')f_1(s,x')\,dx'\right],
\nonumber
\end{multline}
where $W(s)$ denotes the Wronskian of $(f_1)_{[s]}$ and $(f_2)_{[s]}$ and the $(\nu\times \lambda)$-measurable function $\varphi$ on $S\times (a,b)$
is defined by the relation
\[
\varphi(s,x) = -\tilde g''_{[s]} + v(s,x)\tilde g(s,x).
\]
Since the function $s\to W(s)$ on $S$ is $\nu$-measurable, Lemma~\ref{l_meas1} implies that $f$ is $(\nu\times \lambda)$-measurable. As
$f_{[s]}=\sigma_{\tilde g_{[s]}}$ for $\nu$-a.e. $s$, it follows from~(\ref{hsigma}) that $H(s)=L^{f_{[s]}}_{v_{[s]}}$ for $\nu$-a.e. $s$. The Lemma
is proved.
\end{proof}

Let $\nu$ be a measure on a measurable space $S$ and $v$ be as in Lemma~\ref{l_ext1}. Let $f_1$ and $f_2$ be real $(\nu\times\lambda)$-measurable
solutions of~(\ref{main_eq}) on $S$ such that $(f_1)_{[s]}$ and $(f_2)_{[s]}$ are linearly independent for $\nu$-a.e. $s$ (such solutions always
exist by Corollary~\ref{cor_sol}). For any $\nu$-measurable map $\theta$ from $S$ to $[0,\pi)$, we define the real $(\nu\times\lambda)$-measurable
solution $\hat \theta$ of~(\ref{main_eq}) on $S$ by the relation
\begin{equation}\label{hattheta}
\hat\theta(s,x) = f_1(s,x)\cos \theta(s) + f_2(s,x)\sin \theta(s).
\end{equation}
Let $f$ be a real $(\nu\times\lambda)$-measurable solution of~(\ref{main_eq}) on $S$. Then there are $\nu$-a.e. defined real functions $C_1$ and
$C_2$ on $S$ such that the equality
\[
f(s,x) = C_1(s) f_1(s,x) + C_2(s) f_2(s,x)
\]
holds for $\nu$-a.e. $s\in S$ and all $x\in (a,b)$. For $\nu$-a.e. $s$, we have
\[
C_1(s) = \frac{W(f_{[s]},(f_2)_{[s]})}{W((f_1)_{[s]},(f_2)_{[s]})},\quad C_2(s) = \frac{W(f_{[s]},(f_1)_{[s]})}{W((f_2)_{[s]},(f_1)_{[s]})},
\]
and, therefore, both $C_1$ and $C_2$ are $\nu$-measurable on $S$. Let $U\subset \R^2$ be the set of all points of the form $(r\cos \varphi, r\sin
\varphi)$ with $r\geq 0$ and $\varphi\in [0,\pi)$ and let $\Sigma$ denote the intersection of $U$ with the unit circle (in other words, $\Sigma$ is
the set of all points of the form $(\cos \varphi, \sin \varphi)$ with $\varphi\in [0,\pi)$). We define the $\nu$-measurable functions $\tilde C_1$
and $\tilde C_2$ on $S$ by the equalities
\[
\tilde C_1(s) = C_1(s)/C(s),\quad \tilde C_2(s) = C_2(s)/C(s),
\]
where the $\nu$-measurable function $C$ is given by
\[
C(s) =  \left\{
\begin{matrix}
\sqrt{C_1(s)^2 + C_2(s)^2},& s\in U,\\
-\sqrt{C_1(s)^2 + C_2(s)^2},& s\in S\setminus U.
\end{matrix}
\right.
\]
We then have $(\tilde C_1(s), \tilde C_2(s))\in\Sigma$ for $\nu$-a.e. $s$. Let $\chi$ be the map $\varphi\to (\cos\varphi,\sin\varphi)$ from
$[0,\pi)$ to $\Sigma$. Clearly, $\chi$ is a bijection and both $\chi$ and $\chi^{-1}$ are continuous. We now define the $\nu$-measurable function
$\theta$ from $S$ to $[0,\pi)$ by setting
\[
\theta(s) = \chi^{-1}(\tilde C_1(s), \tilde C_2(s)).
\]
We then have
\begin{equation}\label{theta}
f(s,x) = C(s) \hat \theta(s,x)
\end{equation}
for $\nu$-a.e. $s\in S$ and all $x\in (a,b)$. Relation~(\ref{theta}) determines $C$ and $\theta$ uniquely up to $\nu$-equivalence. Indeed, suppose
there are functions $\tilde C$ and $\tilde\theta$ such that~(\ref{theta}) holds with $C$ and $\theta$ replaced with $\tilde C$ and $\tilde\theta$
respectively. Then we have
\[
C(s)\hat \theta(s,x) = \hat C(s)\hat\theta(s,x)
\]
for $\nu$-a.e. $s\in S$ and all $x\in (a,b)$. Since $(f_1)_{[s]}$ and $(f_2)_{[s]}$ are linearly independent for $\nu$-a.e. $s$, it follows that
\[
C(s)\cos\theta(s) = \tilde C(s)\cos\tilde\theta(s),\quad C(s)\sin\theta(s) = \tilde C(s)\sin\tilde\theta(s)
\]
for $\nu$-a.e. $s$. Because both $(\cos\theta(s),\sin\theta(s))$ and $(\cos\tilde\theta(s),\sin\tilde\theta(s))$ belong to $\Sigma$ for $\nu$-a.e.
$s$, this implies that $C(s)=\tilde C(s)$ and $\theta(s)=\tilde \theta(s)$ for $\nu$-a.e. $s$.

By Lemma~\ref{l_ext1}, $L^{\hat\theta_{[s]}}_{v_{[s]}}$ constitute a $\nu$-measurable family of self-adjoint extensions of $L_{v_{[s]}}$ for any
$\nu$-measurable map $\theta$ from $S$ to $[0,\pi)$. Conversely, let $H(s)$ be a $\nu$-measurable family of self-adjoint extensions of $L_{v_{[s]}}$.
Then it follows from Lemma~\ref{l_ext1} and the above arguments concerning representation~(\ref{theta}) that there is a unique (up to
$\nu$-equivalence) $\nu$-measurable map $\theta$ from $S$ to $[0,\pi)$ such that $H(s) = L^{\hat\theta_{[s]}}_{v_{[s]}}$ for $\nu$-a.e. $s$.

We now consider a more general case, when a $(\nu\times\lambda)$-measurable function $v$ on $S\times (a,b)$ is such that both lpc and lcc may hold
for $l_{v_{[s]}}$ at $a$. Let $A_{\mathrm{lc}}$ denote the set of all $s\in S$ such that $v_{[s]}$ is locally square-integrable on $(a,b)$ and lcc
holds for $l_{v_{[s]}}$ at $a$. Let $f_1$ and $f_2$ be as in Corollary~\ref{cor_sol}. For any $c\in (a,b)$, the set of all $s$ such that
\[
\int_a^c (f_1(s,x)^2+f_2(s,x)^2)\,dx < \infty
\]
differs from $A_{\mathrm{lc}}$ by at most a $\nu$-null set. It follows from the Fubini theorem that $A_{\mathrm{lc}}$ is $\nu$-measurable.

Let $f_1$ and $f_2$ be $(\nu\times\lambda)$-measurable solutions of~(\ref{main_eq}) on $A_{\mathrm{lc}}$ such that $(f_1)_{[s]}$ and $(f_2)_{[s]}$
are linearly independent for $\nu$-a.e. $s\in A_{\mathrm{lc}}$. Given a $\nu$-measurable map  $\theta$ from $A_{\mathrm{lc}}$ to $[0,\pi)$, we define
the family $\mathcal L_\theta(s)$ of self-adjoint extensions of $L_{v_{[s]}}$ on $S$ by the relation
\begin{equation}\label{L_family}
\mathcal L_\theta(s) = \left\{
\begin{matrix}
L^{\hat\theta_{[s]}}_{v_{[s]}},& s\in A_{\mathrm{lc}},\\
\bar L_{v_{[s]}},& s\in S\setminus A_{\mathrm{lc}},
\end{matrix}
\right.
\end{equation}
where the solution $\hat\theta$ of~(\ref{main_eq}) on $A_{\mathrm{lc}}$ is given by formula~(\ref{hattheta}) for $s\in A_{\mathrm{lc}}$. The family
$\mathcal L_\theta$ is $\nu$-measurable on $A_{\mathrm{lc}}$ and $S\setminus A_{\mathrm{lc}}$ by Lemmas~\ref{l_ext1} and~\ref{l_v} respectively.
Suppose now that $H(s)$ is a $\nu$-measurable family of self-adjoint extensions of $L_{v_{[s]}}$. Replacing $S$ with $A_{\mathrm{lc}}$ in the above
consideration, we conclude that there is a $\nu$-measurable map  $\theta$ from $A_{\mathrm{lc}}$ to $[0,\pi)$ such that $H(s) =
L^{\hat\theta_{[s]}}_{v_{[s]}}$ for $\nu$-a.e. $s\in A_{\mathrm{lc}}$. Since $L_{v_{[s]}}$ is essentially self-adjoint for $\nu$-a.e. $s\in
S\setminus A_{\mathrm{lc}}$, it follows that $H(s)=\bar L_{v_{[s]}}$ for $\nu$-a.e. $s\in S\setminus A_{\mathrm{lc}}$ and, therefore, $H(s) =
\mathcal L_\theta(s)$ for $\nu$-a.e. $s\in S$. We thus have proved the next theorem.
\begin{theorem} \label{t_meas}
Let $\nu$ be a measure on a measurable space $S$ and $v$ be a $(\nu\times\lambda)$-measurable real function on $S\times (a,b)$ such that $v_{[s]}$ is
locally square-integrable and lpc holds for $l_{v_{[s]}}$ at $b$ for $\nu$-a.e. $s$. Let $A_{\mathrm{lc}}$ be the set of all $s\in S$ such that
$v_{[s]}$ is locally square-integrable on $(a,b)$ and lcc holds for $l_{v_{[s]}}$ at $a$. Let $f_1$ and $f_2$ be $(\nu\times\lambda)$-measurable
solutions of~$(\ref{main_eq})$ on $A_{\mathrm{lc}}$ such that $(f_1)_{[s]}$ and $(f_2)_{[s]}$ are linearly independent for $\nu$-a.e. $s\in
A_{\mathrm{lc}}$. If $\theta$ is a $\nu$-measurable map from $A_{\mathrm{lc}}$ to $[0,\pi)$, then the family $\mathcal L_\theta(s)$ of self-adjoint
extensions of $L_{v_{[s]}}$ defined by~$(\ref{L_family})$, where $\hat \theta$ is given by~$(\ref{hattheta})$, is $\nu$-measurable. Conversely, if
$H(s)$ is a $\nu$-measurable family of self-adjoint extensions of $L_{v_{[s]}}$, then there is a unique up to $\nu$-equivalence $\nu$-measurable map
$\theta$ from $A_{\mathrm{lc}}$ to $[0,\pi)$ such that $H(s) = \mathcal L_\theta(s)$ for $\nu$-a.e. $s$.
\end{theorem}

\section{Self-adjoint extensions of the three-dimensional Aharonov--Bohm Hamiltonian}
\label{s8}

The Hamiltonian for an electron moving in the magnetic field of an infinitely thin solenoid is formally given by the differential expression
\begin{equation}\label{diff_expr}
\frac{\hbar^2}{2m_{\mathrm e}}\left(i\nabla + \frac{e}{\hbar c}\mathbf A\right)^2,
\end{equation}
where $e$ and $m_{\mathrm e}$ are the electron charge and mass respectively, $c$ is the velocity of light, and the vector potential $\mathbf
A=(A^1,A^2,A^3)$ has the form
\[
A^1(x,y,z) = -\frac{\Phi y}{2\pi(x^2+y^2)},\quad A^2(x,y,z) = \frac{\Phi x}{2\pi(x^2+y^2)},\quad A^3(x,y,z)=0.
\]
Here, $\Phi$ is the flux of the magnetic field through the solenoid. The vector potential $\mathbf A$ is smooth outside the $z$-axis
$Z=\{(x,y,z)\in\R^3 : x=y=0\}$. Hence, (\ref{diff_expr}) naturally determines an operator $\check H$ on the space $C_0^\infty(\R^3\setminus Z)$ of
smooth functions on $\R^3$ with compact support contained in $\R^3\setminus Z$,
\begin{multline}
\check H \Psi = \left(i\nabla + \frac{e}{\hbar c}\mathbf A\right)^2\Psi = \\= \left(-\Delta +\frac{2i\phi}{x^2+y^2}(y\partial_x-x\partial_y)+
\frac{\phi^2}{x^2+y^2}\right)\Psi,\quad \Psi\in C_0^\infty(\R^3\setminus Z), \label{checkH}
\end{multline}
where
\[
\phi = -\frac{e\Phi}{2\pi\hbar c}
\]
(to simplify notation, we have dropped the factor $\hbar^2/2m_{\mathrm e}$ in~(\ref{diff_expr})). Lifting $\check H$ to $\Lambda$-equivalence
classes, where $\Lambda$ is the Lebesgue measure on $\R^3$, yields a symmetric operator $H$ in $L^2(\R^3)$:
\begin{align}
& D_H = \left\{ [\Psi]_\Lambda : \Psi\in C_0^\infty(\R^3\setminus Z)\right\}, \nonumber\\
& H[\Psi]_\Lambda = [\check H \Psi]_\Lambda. \nonumber
\end{align}
Let $\mathcal G$ be the Abelian group of linear operators in $\R^3$ generated by translations along the $z$-axis and rotations around the $z$-axis.
Given $G\in \mathcal G$, we denote by $U_G$ the unitary operator in $L^2(R^3)$ taking $[\Psi]_\Lambda$ to $[\Psi\circ G^{-1}]_\Lambda$ for any
square-integrable function $\Psi$ on $\R^3$. It is straightforward to check that $H$ commutes with $U_G$ for any $G\in \mathcal G$. We shall see that
$H$ is not essentially self adjoint. Hence, there are different quantum dynamics that can be associated with differential
expression~(\ref{diff_expr}) via constructing different self-adjoint extensions of $H$. In this section, we shall describe all self-adjoint
extensions of $H$ commuting with $U_G$ for any $G\in \mathcal G$.

We begin by constructing an exact diagonalization for the operators $U_G$. Let $\mu$ be the counting measure on $\Z$, which assigns to each set of
integers the number of points in the set. We define the measure $\nu$ on $S=\Z\times \R$ by setting $\nu=\mu\times\lambda$, where $\lambda$ is the
Lebesgue measure on $\R$. For any $\nu$-integrable $f$, we have
\[
\int_{S} f(m,p)\,d\nu(m,p) = \sum_{m\in\Z} \int_{-\infty}^\infty f(m,p)\,dp.
\]
For $\Psi\in C_0^\infty(\R^3)$, let the function $\tilde\Psi$ on $S\times \R_+$, where $\R_+=(0,\infty)$, be defined by the relation
\begin{equation}\label{tilde}
\tilde\Psi(s,r) = \frac{\sqrt{r}}{2\pi}\int_{-\infty}^\infty dz\int_{0}^{2\pi}d\varphi \Psi(r\cos\varphi,r\sin\varphi,z)e^{ipz+im\varphi},
\quad s=(m,p)\in S.
\end{equation}
Let $\mathfrak h = L^2(\R_+)$. For $\Psi\in C_0^\infty(\R^3)$, we define the map $\hat\Psi$ from $S$ to $\mathfrak h$ by setting
\begin{equation}\label{hatPsi}
\hat\Psi(s) = [\tilde\Psi_{[s]}]_{\lambda_+},
\end{equation}
where $\lambda_+$ is the restriction to $\R_+$ of the Lebesgue measure $\lambda$ on $\R$ and $\tilde\Psi_{[s]}\in C_0^\infty(\R_+)$ denotes, as in
Sec.~\ref{s_meas}, the partial function on $\R_+$ determined by $\tilde\Psi(s,r)$ for fixed $s$,
\[
\tilde\Psi_{[s]}(r) = \tilde\Psi(s,r),\quad r\in \R_+.
\]

The next lemma follows easily from the Fubini theorem and the unitarity of the Fourier transformation and the Fourier series expansion.

\begin{lemma}\label{l_unitary}
There is a unique unitary operator $V\colon L^2(\R^3)\to L^2(S,\mathfrak h,\nu)$ such that
\[
V[\Psi]_\Lambda = [\hat\Psi]_\nu,\quad \Psi\in C_0^\infty(\R^3).
\]
\end{lemma}

Given $\alpha,\beta\in\R$, let the function $g_{\alpha,\beta}$ on $S$ be given by
\[
g_{\alpha,\beta}(m,p) = e^{i\alpha m + i\beta p}.
\]
If $G\in \mathcal G$ is the composition of the rotation by the angle $\alpha$ around $z$-axis and the translation by $\beta$ along $z$-axis, then it
is easy to see that
\[
VU_G V^{-1} = \mathcal T_{g_{\alpha,\beta}},
\]
where $\mathcal T_{g_{\alpha,\beta}}$ is the operator of multiplication by $g_{\alpha,\beta}$ in $L^2(S,\mathfrak h,\nu)$. We now show that
$\{g_{\alpha,\beta}\}_{(\alpha,\beta)\in\mathbb Q^2}$, where $\mathbb Q$ is the set of rational numbers, is a $\nu$-separating family of functions on
$S$. Suppose $(m,p)$ and $(m',p')$ are such that $g_{\alpha,\beta}(m,p)=g_{\alpha,\beta}(m',p')$ for all $(\alpha,\beta)\in \mathbb Q^2$. Then we
have
\[
e^{i\alpha(m-m') +i\beta(p-p')}=1
\]
for all $(\alpha,\beta)\in \mathbb Q^2$. Since $\mathbb Q^2$ is dense in $\R^2$, it follows that $m=m'$ and $p=p'$, i.e., the family
$\{g_{\alpha,\beta}\}_{(\alpha,\beta)\in\mathbb Q^2}$ is $\nu$-separating. Theorem~\ref{t0a} now implies that $(S,\mathfrak h,\nu,V)$ is an exact
diagonalization for $U_G$.

It easily follows from (\ref{checkH}) that
\begin{multline}\label{polar}
(\check H\Psi)(r\cos \varphi,r\sin\varphi,z) =\\= \left(-\partial_z^2-
\partial_r^2-\frac{1}{r}\partial_r - \frac{1}{r^2}(\partial_\varphi^2+2i\phi\partial_\varphi-\phi^2)\right)
F_\Psi(\varphi,z,r)
\end{multline}
for any $\Psi\in C^\infty_0(\R^3\setminus Z)$, where $F_\Psi$ is the smooth function on $\R\times\R\times\R_+$ which represents $\Psi$ in the
cylindrical coordinates,
\[
F_\Psi(\varphi,z,r) = \Psi(r\cos \varphi,r\sin\varphi,z).
\]
Substituting (\ref{polar}) in~(\ref{tilde}) and integrating by parts yields
\begin{equation}\label{checkh}
\widetilde{\check H\Psi}(s,r) = (\check h_{m-\phi}\tilde \Psi_{[s]})(r) + p^2\tilde\Psi(s,r),\quad s=(m,p)\in S,
\end{equation}
where the operator $\check h_\kappa$ from $C^\infty_0(\R_+)$ to itself is given by
\[
(\check h_\kappa\psi)(r) = -\psi''(r) + \frac{\kappa^2-1/4}{r^2}\psi(r),\quad \psi\in C^\infty_0(\R_+),
\]
for any $\kappa\in\R$. Let $h_\kappa$ denote the operator in $\mathfrak h$ obtained by lifting $\check h_\kappa$ to $\lambda$-equivalence classes,
\begin{align}
& D_{h_\kappa} = D_0=\left\{ [\psi]_{\lambda_+} : \psi\in C_0^\infty(\R_+)\right\}, \nonumber\\
& h_\kappa[\psi]_{\lambda_+} = [\check h_\kappa \psi]_{\lambda_+}. \nonumber
\end{align}
It follows from~(\ref{hatPsi}) and~(\ref{checkh}) that
\begin{equation}\label{amp}
\widehat{\check H\Psi}(s) = a(s) \hat\Psi(s),\quad s\in S,
\end{equation}
where
\begin{equation}\label{amp_d}
a(m,p) = h_{m-\phi}+p^2 1_{\mathfrak h}
\end{equation}
for any $(m,p)\in S$ and $1_{\mathfrak h}$ is the identity operator in $\mathfrak h$. Let the function $v$ on $S\times \R_+$ be defined by the
relation
\begin{equation}\label{vsr}
v(s,r) = \frac{(m-\phi)^2-1/4}{r^2},\quad s=(m,p)\in S.
\end{equation}
Clearly, $v$ is $(\nu\times\lambda)$-measurable on $S\times \R_+$ and, in the notation of~Sec.~\ref{s_meas}, we have
\begin{equation}\label{lvs}
L_{v_{[s]}} = h_{m-\phi}
\end{equation}
for any $s=(m,p)\in S$. Hence, Lemma~\ref{l_v} implies that the family $(m,p)\to h_{m-\phi}$ of operators on $S$ is $\nu$-measurable and
$\nu$-regular with respect to $D_0$. Since $(m,p)\to p^2 1_{\mathfrak h}$ is a $\nu$-measurable family of bounded operators on $S$, it follows
from~(\ref{amp_d}) that the family $a(m,p)$ is also $\nu$-measurable and $\nu$-regular with respect to $D_0$. Fix a nonzero function $\chi\in
C^\infty_0(\R)$. For $\psi\in C^\infty_0(\R_+)$ and $n\in \Z$, let $\Psi_{n,\psi}\in C^\infty_0(\R^3\setminus Z)$ be the function whose cylindrical
coordinate representation has the form
\[
F_{\Psi_{n,\psi}}(\varphi,z,r) = \frac{1}{\sqrt{r}}e^{-in\varphi}\chi(z)\psi(r).
\]
Then we have
\begin{equation}\label{hat}
\hat \Psi_{n,\psi}(m,p) = \delta_{m,n}\,\mathcal F\chi(p)\,[\psi]_{\lambda_+},\quad (m,p)\in S,
\end{equation}
where $\mathcal F\chi$ is the Fourier transform of $\chi$,
\[
\mathcal F\chi(p) = \int \chi(z)e^{izp}\,dz,
\]
and $\delta_{m,n}$ is, as usual, equal to $1$ for $m=n$ and equal to $0$ for $m\neq n$. Since $\mathcal F\chi$ admits the analytic continuation to
$\C$, the set of its zeros has Lebesgue measure zero. It follows that the functions $g_n(m,p)= \delta_{m,n}\,\mathcal F\chi(p)$ constitute a
$\nu$-nonvanishing sequence of square-integrable functions o $S$ (see the paragraph before Lemma~\ref{l6}). By~(\ref{hat}), the $\nu$-equivalence
class of the map $s\to g_n(s)[\psi]_{\lambda,\R_+}$ is equal to $V[\Psi_{n,\psi}]_\Lambda$ and, therefore, belongs to $V(D_H)$ for any $\psi\in
C^\infty_0(\R_+)$. Hence, statement~2 of Lemma~\ref{l6} implies that the family $a(s)$ is compatible with $V(D_H)$.

By~(\ref{amp}), we have
\begin{equation}\label{100}
a(s)\xi(s) = \eta(s)
\end{equation}
for $\nu$-a.e. $s\in S$, whenever $\xi$ and $\eta$ are $\nu$-measurable maps from $S$ to $\mathfrak h$ such that
\begin{equation}\label{101}
[\xi]_\nu = V[\Psi]_\Lambda,\quad [\eta]_\nu = VH[\Psi]_\Lambda
\end{equation}
for some $\Psi\in C^\infty_0(\R^3\setminus Z)$. Taking~(\ref{100}), (\ref{101}), the compatibility of $a(s)$ with $V(D_H)$, and the exactness of the
diagonalization $(S,\mathfrak h,\nu,V)$ for $U_G$ into account and applying Theorem~\ref{t1}, we arrive at the next result.
\begin{lemma}\label{l100}
Let $\tilde a(s)$ be a $\nu$-measurable family of self-adjoint extensions of $a(s)$ on $S$. Then the operator
\begin{equation}\label{tilH}
\tilde H = V^{-1}\int^\oplus \tilde a(s) \,d\nu(s)\, V
\end{equation}
is a self-adjoint extension of $H$ commuting with $U_G$ for all $G\in \mathcal G$. Conversely, if $\tilde H$ is a self-adjoint extension of $H$
commuting with $U_G$ for all $G\in \mathcal G$, then there is a unique (up to $\nu$-equivalence) $\nu$-measurable family $\tilde a(s)$ of
self-adjoint extensions of $a(s)$ on $S$ such that $(\ref{tilH})$ holds.
\end{lemma}

It follows from~(\ref{amp_d}) that $\tilde a(s)$ is a $\nu$-measurable family of self-adjoint extensions of $a(s)$ if and only if
\[
\tilde a(m,p) = \tilde h(m,p) + p^2 1_{\mathfrak h}
\]
for $\nu$-a.e. $(m,p)\in S$, where $\tilde h(s)$ is a $\nu$-measurable family of operators on $S$ such that $\tilde h(m,p)$ is a self-adjoint
extension of $h_{m-\phi}$ for $\nu$-a.e. $(m,p)\in S$. In view of Lemma~\ref{l100}, this implies that the problem of describing self-adjoint
extensions of $H$ reduces to describing all such families $\tilde h(s)$.

As in Sec.~\ref{s_meas}, let $\mathcal D$ denote the space of all absolutely continuous complex functions on $\R_+$ whose first derivative is also
absolutely continuous. For $\kappa\in \R$, let $l_\kappa$ be the linear map from $\mathcal D$ to the space of complex-valued $\lambda_+$-equivalence
classes taking $\psi\in \mathcal D$ to the $\lambda_+$-equivalence class of the map
\[
r\to -\psi''(r) + \frac{\kappa^2-1/4}{r^2}\psi(r).
\]
Let the subspace $\mathcal D_\kappa$ of $\mathcal D$ and the subspace $D_\kappa$ of $\mathfrak h$ be defined by the relations
\begin{align}
&\mathcal D_\kappa =
\{\psi\in \mathcal D : [\psi]_{\lambda_+} \mbox{ and } l_\kappa \psi \mbox{ are both in } \mathfrak h\},\nonumber \\
& D_\kappa = \{ [\psi]_{\lambda_+} : \psi\in \mathcal D_\kappa\}.\nonumber
\end{align}
In the notation of Sec.~\ref{s_meas}, we have
\begin{equation}\label{equalities}
h_\kappa = L_{q_\kappa},\quad l_\kappa = l_{q_\kappa},\quad \mathcal D_\kappa = \mathcal D_{q_\kappa},\quad D_\kappa = D_{q_\kappa},
\end{equation}
where the function $q_\kappa$ on $\R_+$ is given by
\[
q_\kappa(r) = \frac{\kappa^2-1/4}{r^2}.
\]
Hence, the adjoint $h_\kappa^*$ of $h_\kappa$ has the form
\begin{align}
& D_{h_\kappa^*} = D_\kappa,\nonumber\\
& h_\kappa^* [\psi]_{\lambda_+} = l_\kappa \psi, \quad \psi\in \mathcal D_\kappa.\nonumber
\end{align}
The equation
\begin{equation}\label{201}
l_\kappa \psi =0
\end{equation}
has two linearly independent solutions
\begin{align}
& \psi^{(1)}_\kappa(r) = r^{1/2+\kappa},\quad \psi^{(2)}_\kappa(r) = r^{1/2-\kappa},\quad \kappa\neq 0,\nonumber\\
& \psi^{(1)}_\kappa(r) = r^{1/2},\quad \psi^{(2)}_\kappa(r) = r^{1/2}\ln r,\quad \kappa= 0.\nonumber
\end{align}
Hence, lpc holds at $r=\infty$ for all $\kappa$, while lpc holds at $r=0$ for $|\kappa|\geq 1$ and lcc holds at $r=0$ for $|\kappa|<1$. This implies
that $h_\kappa$ is essentially self-adjoint for $|\kappa|\geq 1$ and its unique self-adjoint extension is $h_\kappa^*$. For $\vartheta\in\R$, let the
solution $\psi_{\kappa,\vartheta}$ of~(\ref{201}) be given by
\begin{equation}\label{psikt}
\psi_{\kappa,\vartheta}(r) = \psi^{(1)}_\kappa(r) \cos\vartheta + \psi^{(2)}_\kappa(r) \sin \vartheta.
\end{equation}
For $|\kappa|<1$, we define the self-adjoint extension $h_{\kappa,\vartheta}$ of $h_\kappa$ by setting
\begin{equation}\label{hkt}
h_{\kappa,\vartheta}=L_{q_\kappa}^{\psi_{\kappa,\vartheta}},
\end{equation}
i.e., $h_{\kappa,\vartheta}$ is the restriction of $h_\kappa^*$ to
\begin{equation}
D_{h_{\kappa,\vartheta}} =\left \{[\psi]_{\lambda_+} : \psi\in \mathcal D_\kappa \mbox{ and } \lim_{r\downarrow 0}
W(\psi,\psi_{\kappa,\vartheta})(r)=0\right\},
\nonumber
\end{equation}
where the Wronskian $W$ is given by~(\ref{wronskian}).

For each $\phi\in \R$, there is a unique $m(\phi)\in\Z$ such that $m(\phi)-\phi \in (-1,0]$ (note that $m(\phi)=\phi$ for $\phi\in\Z$). The operator
$h_{m-\phi}$ is not essentially self-adjoint for $m=m(\phi)$ if $phi\in \Z$ and for $m=m(\phi),m(\phi)+1$ if $\phi\notin\Z$. Hence, families of
self-adjoint extensions of $h_{m-\phi}$ are defined differently for $\phi\in\Z$ and $\phi\notin\Z$.

\par\medskip\noindent
1. Let $\phi\notin \Z$ and let $\tau_1$ and $\tau_2$ be $\lambda$-measurable maps from $\R$ to $[0,\pi)$. We define the family $\tilde
h_{\tau_1,\tau_2}(s)$ of self-adjoint operators on $S$ by setting
\begin{equation}\label{htau12}
\tilde
h_{\tau_1,\tau_2}(m,p) = \left\{
\begin{matrix}
h_{m-\phi}^*,& m<m(\phi)\mbox{ or } m>m(\phi)+1,\\
h_{m-\phi,\tau_1(p)},& m=m(\phi),\\
h_{m-\phi,\tau_2(p)},& m=m(\phi)+1.
\end{matrix}
\right.
\end{equation}

\par\medskip\noindent
2. Let $\phi\in \Z$ and let $\tau$ be a $\lambda$-measurable map from $\R$ to $[0,\pi)$. We define the family $\tilde h_{\tau}(s)$ of self-adjoint
operators on $S$ by setting
\begin{equation}\label{htau}
\tilde
h_{\tau}(m,p) = \left\{
\begin{matrix}
h_{m-\phi}^*,& m\neq \phi,\\
h_{\,0,\tau(p)},& m=\phi.
\end{matrix}
\right.
\end{equation}

\begin{theorem}\label{t2}
$\,$
\par\noindent $1.$ Let $\phi\notin \Z$. Suppose $\tau_1$ and $\tau_2$ are $\lambda$-measurable maps from $\R$ to $[0,\pi)$. Then the family $\tilde
h_{\tau_1,\tau_2}(s)$ on $S$ of self-adjoint operators in $\mathfrak h$ defined by~$(\ref{htau12})$ is $\nu$-measurable and
\begin{equation}\label{301}
\tilde H = V^{-1}\int_S^\oplus (\tilde h_{\tau_1,\tau_2}(m,p)+ p^2 1_{\mathfrak h})\,d\nu(m,p)\,\, V
\end{equation}
is a self-adjoint extension of $H$. Conversely, for any self-adjoint extension $\tilde H$ of $H$, there are unique (up to $\lambda$-equivalence)
$\lambda$-measurable maps $\tau_1$ and $\tau_2$ from $\R$ to $[0,\pi)$ such that $(\ref{301})$ holds.

\par\medskip\noindent $2.$ Let $\phi\in \Z$. Suppose $\tau$ is a $\lambda$-measurable map from $\R$ to $[0,\pi)$. Then the family $\tilde
h_{\tau}(s)$ on $S$ of self-adjoint operators in $\mathfrak h$ defined by~$(\ref{htau})$ is $\nu$-measurable and
\begin{equation}\label{302}
\tilde H = V^{-1}\int_S^\oplus (\tilde h_{\tau}(m,p)+ p^2 1_{\mathfrak h})\,d\nu(m,p)\,\, V
\end{equation}
is a self-adjoint extension of $H$. Conversely, for any self-adjoint extension $\tilde H$ of $H$, there is a unique (up to $\lambda$-equivalence)
$\lambda$-measurable map $\tau$ from $\R$ to $[0,\pi)$ such that $(\ref{302})$ holds.
\end{theorem}

\begin{proof}
1. Let the functions $f_1$ and $f_2$ on $S\times \R_+$ be defined by the formulas
\[
f_1(s,r) = \psi^{(1)}_{m-\phi}(r),\quad f_2(s,r) = \psi^{(2)}_{m-\phi}(r),\quad s=(m,p)\in S.
\]
In view of~(\ref{equalities}), the functions $f_1$ and $f_2$ are real linearly independent solutions of~(\ref{main_eq}) for $v$ given by~(\ref{vsr})
because $v_{[s]}= q_{m-\phi}$ for all $s=(m,p)\in S$. The set $A_{\mathrm{lc}}$ of all $s\in S$ such that lcc holds at $r=0$ for $l_{v_{[s]}}$ has
the form
\[
A_{\mathrm{lc}} = \{m(\phi),m(\phi)+1\}\times \R.
\]
If $\tau_1$ and $\tau_2$ are $\lambda$-measurable maps from $\R$ to $[0,\pi)$, then (\ref{psikt}), (\ref{hkt}), and (\ref{htau12}) imply that
\begin{equation}\label{400}
\tilde h_{\tau_1,\tau_2}(s) = \mathcal L_\theta(s)
\end{equation}
for $\nu$-a.e. $s\in S$, where $\mathcal L_\theta(s)$ is given by~(\ref{L_family}) and the $(\nu\times \lambda)$-measurable map $\theta$ from
$A_{\mathrm{lc}}$ to $[0,\pi)$ is defined by the relations
\[
\theta(m(\phi),p)=\tau_1(p),\quad \theta(m(\phi)+1,p)=\tau_2(p).
\]
Theorem~\ref{t_meas} now implies that $\tilde h_{\tau_1,\tau_2}(s)$ is a $\nu$-measurable family of operators on $S$. Since $\tilde
h_{\tau_1,\tau_2}(m,p)$ is a self-adjoint extension of $h_{m-\phi}$ for all $(m,p)\in S$, the operators
\begin{equation}\label{401}
\tilde a(m,p) = \tilde h_{\tau_1,\tau_2}(m,p) + p^2 1_{\mathfrak h}
\end{equation}
constitute a $\nu$-measurable family of self-adjoint extensions of $a(m,p)$. Lemma~\ref{l100} hence implies that the operator $\tilde H$ defined
by~(\ref{301}) is a self-adjoint extension of $H$.

Conversely, let $\tilde H$ be a self-adjoint extension of $H$. By Lemma~\ref{l100}, there is a unique (up to $\nu$-equivalence) $\nu$-measurable
family $\tilde a(s)$ of self-adjoint extensions of $a(s)$ on $S$ such that (\ref{tilH}) holds. Hence, the operators $\tilde a(m,p)-p^2 1_{\mathfrak
h}$ constitute a $\nu$-measurable family of self-adjoint extensions of $h_{m-\phi}$. In view of~(\ref{lvs}), Theorem~\ref{t_meas} implies that there
is a unique (up to $\nu$-equivalence) $\nu$-measurable map $\theta$ from $A_{\mathrm{lc}}$ to $[0,\pi)$ such that
\begin{equation}\label{402}
\tilde a(m,p)-p^2 1_{\mathfrak h} = \mathcal L_\theta(m,p).
\end{equation}
Let $\tau_1$ and $\tau_2$ denote the maps $p\to \theta(m(\phi),p)$ and $p\to \theta(m(\phi)+1,p)$ respectively. Then both $\tau_1$ and $\tau_2$ are
$\nu$-measurable maps from $\R$ to $[0,\pi)$ and~(\ref{400}) holds for $\nu$-a.e. $s\in S$. Now substituting~(\ref{400}) in~(\ref{402})
yields~(\ref{401}), and substituting~(\ref{401}) in~(\ref{tilH}) yields~(\ref{301}). It remains to note that~(\ref{301}) determines $\tau_1$ and
$\tau_2$ uniquely up to $\lambda$-equivalence.

\par\medskip\noindent
2. The proof of statement~2 is the same as the proof of statement~1 with the only difference that we have
\[
A_{\mathrm{lc}} = \{\phi\}\times \R
\]
in this case. The theorem is proved.

\end{proof}

\appendix

\section{Some measurability questions}

\begin{lemma}\label{l_meas}
Let $\nu$ be a measure on a measurable space $S$ and $\lambda$ be a standard measure on a measurable space $T$. Let $g$ be a $(\nu\times
\lambda)$-a.e. defined complex-valued function on $S\times T$. Let $h$ be a $\nu$-a.e. defined map from $S$ to $L^2(T,\lambda)$ such that $h(s)=
[g_{[s]}]_\lambda$ for $\nu$-a.e. $s$. Then $g$ is $(\nu\times\lambda)$-measurable if and only if $h$ is $\nu$-measurable.
\end{lemma}
\begin{proof}
Let $g$ be $(\nu\times\lambda)$-measurable. We have to show that the function $s\to ([f]_\lambda, h(s))$ is $\nu$-measurable for any
square-integrable function $f$ on $T$. Then $|g|^2$ is a $(\nu\times\lambda)$-measurable map from $S\times T$ to the extended positive semi-axis, and
the Fubini theorem implies that $s\to \int_T |g(s,t)|^2\,d\lambda(t)$ is a $\nu$-measurable function on $S$. For $N=1,2,\ldots$, let $A_N$ be the set
of all $s\in S$ such that $\int_T |g(s,t)|^2\,d\lambda(t)\leq N$. Then the function $(s,t)\to g(s,t)\chi_{A_N}(s) \bar f(t)$, where $\chi_{A_N}$ is
the characteristic function of $A_N$, is $(\nu\times\lambda)$-integrable, and it follows from the Fubini theorem that the function $s\to
\chi_{A_N}(s)\int_T \bar f(t)g(s,t)\,d\lambda(t)$ is $\nu$-measurable. This means that the function $s\to \chi_{A_N}(s) ([f]_\lambda, h(s))$ is
$\nu$-measurable. Since $S\setminus \bigcup_N A_N$ is a $\nu$-null set, it follows that the function $s\to ([f]_\lambda, h(s))$ is $\nu$-measurable.

Conversely, let $h$ be $\nu$-measurable. We first suppose that $h$ is square-integrable. Since $\lambda$ is standard, $L^2(T,\lambda)$ is separable.
Let $e_1,e_2,\ldots$ be a sequence of $\lambda$-measurable functions on $T$ such that $[e_1]_\lambda,[e_2]_\lambda,\ldots$ is a basis in
$L^2(T,\lambda)$. For each $n=1,2,\ldots$, we set
\begin{equation}\label{anbn}
a_n(s) = \int_T \bar e_n(t) g(s,t)\,d\lambda(t),\quad b_n(s,t) = \sum_{j=1}^n a_j(s) e_j(t).
\end{equation}
Since $h$ is $\nu$-measurable, all $a_n$ are $\nu$-measurable and, therefore, all $b_n$ are $(\nu\times\lambda)$-measurable. For $\nu$-a.e. $s$, we
have $\|h(s)\|^2 = \sum_{n=1}^\infty |a_n(s)|^2$, and it follows from the monotone convergence theorem that
\begin{equation}\label{sum_int}
\int_S \|h(s)\|^2 \,d\nu(s) = \sum_{n=1}^\infty \int_S |a_n(s)|^2\, d\nu(s).
\end{equation}
This implies that all $b_n$ are square-integrable and
\[
\int_{S\times T} |b_n(s,t)-b_m(s,t)|^2\,d(\nu\times\lambda)(s,t) = \sum_{j=m+1}^n \int_S |a_j(s)|^2\, d\nu(s),\quad n\geq m.
\]
In view of~(\ref{sum_int}), it follows that $[b_n]_{\nu\times\lambda}$ is a Cauchy sequence in $L^2(S\times T, \nu\times\lambda)$ and, therefore, we
can choose a subsequence $b_{n_k}$ that converges $(\nu\times\lambda)$-a.e. to some $(\nu\times\lambda)$-measurable function $\tilde g$. On the other
hand, $[b_{n_k}(s,\cdot)]_\lambda$ converge to $h(s)$ in $L^2(T,\lambda)$ for $\nu$-a.e. $s$. Hence, for $\nu$-a.e. $s$, there is a subsequence of
$b_{n_k}(s,\cdot)$ that converges $\lambda$-a.e. to $g(s,\cdot)$. This means that $g$ and $\tilde g$ coincide $(\nu\times\lambda)$-a.e. and,
therefore, $g$ is $(\nu\times\lambda)$-measurable. In the general case, we denote by $A_N$ the set of all $s\in S$ such that $\|h(s)\|^2\leq N$. Then
the map $s\to \chi_{A_N}(s)h(s)$ is square-integrable and by the above, the function $(s,t)\to \chi_{A_N}(s) g(s,t)$ is
$(\nu\times\lambda)$-measurable. This implies that $g$ is $(\nu\times\lambda)$-measurable. The lemma is proved.
\end{proof}

\begin{lemma}\label{l_meas1}
Let $\nu$ be a measure on a measurable spaces $S$ and $\lambda$ be the Lebesgue measure on an interval $(a,b)$. Let $f$ be a
$(\nu\times\lambda)$-measurable complex-valued function on $S\times (a,b)$ such that $f_{[s]}$ is locally $\lambda$-integrable for $\nu$-a.e. $s$.
Then for any $x_0\in (a,b)$, the function
\[
g(s,x) = \int_{x_0}^x f(s,x')\,dx'
\]
on $S\times (a,b)$ is $(\nu\times\lambda)$-measurable. If $f_{[s]}$ is left $\lambda$-integrable for $\nu$-a.e. $s$, then the statement also holds
for $x_0=a$.
\end{lemma}
\begin{proof}
For each $N=1,2,\ldots$, we choose a partition of $(a,b)$ into measurable sets $A^{N}_1,\ldots,A^{N}_{k_N}$ such that the diameter of every $A^{N}_j$
is less than $1/N$. For each $j=1,\ldots,k_N$, we choose a point $x^{N}_j\in A^{N}_j$ and set
\[
g_{N}(s,x) = \sum_{j=1}^{k_N} \chi_{A^{N}_j}(x) g(s,x^{N}_j),
\]
where $\chi_{A^{N}_j}$ is the characteristic function of $A^N_j$. By the Fubini theorem, the function $s\to g(s,x)$ on $S$ is $\nu$-measurable for
any $x\in (a,b)$ and, therefore, $g_{N}$ are $(\nu\times\lambda)$-measurable for all $N$. Since $g_{[s]}$ is continuous on $(a,b)$ for $\nu$-a.e.
$s$, the sequence $(g_{N})_{[s]}$ converges pointwise to $g_{[s]}$ for $\nu$-a.e. $s$. This implies that $g_{N}$ converge $(\nu\times\lambda)$-a.e.
to $g$ and, hence, $g$ is $(\nu\times\lambda)$-measurable. The lemma is proved.
\end{proof}

\end{document}